\PassOptionsToPackage{dvipsnames,svgnames,x11names}{xcolor}

\documentclass[acmsmall,nonacm]{acmart}

\usepackage{algpseudocode}
\usepackage{amsmath}

\usepackage{amsthm}
\usepackage{array}
\usepackage{caption}
\usepackage{subcaption}
\usepackage{xfrac}
\usepackage{dsfont}
\usepackage{empheq}
\usepackage{hyperref}
\usepackage{listings}
\usepackage{multicol}
\usepackage{paralist}
\usepackage{relsize}
\usepackage{stmaryrd}
\usepackage{tabularx}
\usepackage{tikz}
\usetikzlibrary{fit,calc,automata,trees,positioning,arrows,chains,shapes.geometric,decorations.pathreplacing,decorations.pathmorphing,shapes,matrix,shapes.symbols}
\usepackage{url}
\usepackage[]{units}
\usepackage{verbatim}
\usepackage{wasysym}

\allowdisplaybreaks

\theoremstyle{remark}
\newtheorem{remark}{Remark}
\newtheorem*{notation}{Notation}

\usepackage{thmtools}
\usepackage{thm-restate}

\algnewcommand\algorithmicswitch{\textbf{match}}
\algnewcommand\algorithmiccase{\textbf{with}}
\algnewcommand\algorithmicassert{\texttt{assert}}
\algnewcommand\Assert[1]{\State \algorithmicassert(#1)}%
\algdef{SE}[SWITCH]{Switch}{EndSwitch}[1]{\algorithmicswitch\ #1\ }{\algorithmicend\ \algorithmicswitch}%
\algdef{SE}[CASE]{Case}{EndCase}[1]{\algorithmiccase\ #1\ :}{\algorithmicend\ \algorithmiccase}%
\algtext*{EndSwitch}%
\algtext*{EndCase}%

\newcommand{\pfun}{\mathfrak{p}}
\newcommand{\Perturbations}{\mathtt{P}}
\newcommand{\psem}[1]{\langle\!\langle#1\rangle\!\rangle}

\newcommand{\logicName}{Robustness Temporal Logic}
\newcommand{\logicShort}{RobTL}

\newcommand{\cinterval}{\mathtt{CI}}

\newcommand{\logicSymbol}{\mathtt{L}}

\newcommand{\system}{\mathbf{s}}

\newcommand{\esp}{\mathtt{exp}}
\newcommand{\espClass}{\mathtt{Exp}}

\newcommand{\sx}[1]{<^{#1}}
\newcommand{\dx}[1]{>^{#1}}

\newcommand{\join}{\,\mathtt{max}\,}
\newcommand{\meet}{\,\mathtt{min}\,}
\newcommand{\eventually}[1]{\mathtt{E}^{#1}\,}
\newcommand{\always}[1]{\mathtt{A}^{#1}\,}
\newcommand{\until}[1]{\,\mathtt{U}^{#1}\,}

\newcommand{\cstep}{\mathsf{step}}

\newcommand{\fevent}[2]{\Diamond^{#1} #2}
\newcommand{\fglob}[2]{\Box^{#1} #2}
\newcommand{\funtil}[3]{#1~\mathcal{U}^{#2}~#3}

\newcommand{\nats}{\mathbb{N}}
\newcommand{\borel}{{\mathcal B}}
\newcommand{\reals}{\mathbb{R}}

\newcommand{\rv}{\textbf{X}}

\newcommand{\sat}[3]{\llbracket #1 \rrbracket_{#2}^{#3}}

\newcommand{\D}{\mathcal{D}}
\newcommand{\ds}{\mathbf{d}}

\newcommand{\mesA}{\mathbb{A}}
\newcommand{\mesB}{\mathbb{B}}

\newcommand{\mesD}{\mathbb{D}}

\newcommand{\distrib}{\Pi}

\newcommand{\Var}{\mathcal{V}}

\DeclareMathOperator{\Wasserstein}{\mathbf{W}}

\newcommand{\wact}[2]{(#1\rightarrow #2)}

\newcommand{\dirac}{\delta}

\newcommand{\traccione}{evolution sequence}
\newcommand{\tracciones}{evolution sequences}
\newcommand{\Traccione}{Evolution sequence}
\newcommand{\Tracciones}{Evolution sequences}

\newcommand{\dataspace}{data space}

\newcommand{\Dataspace}{Data space}

\newcommand{\datastate}{data state}
\newcommand{\datastates}{data states}
\newcommand{\Datastate}{Data state}

\newcommand{\ES}{\mathcal{S}}

\newcommand{\dd}{\,\mathfrak{d}\,}

\newcommand{\F}{\Sigma}

\newcommand{\f}{\mathtt{f}}

\newcommand{\h}{\mathfrak{h}}

\newcommand{\p}{\mathtt{p}}

\newcommand{\W}{\mathfrak{W}}
\newcommand{\w}{\mathfrak{w}}

\newcommand{\nat}{\mathbb{N}} 
\newcommand{\real}{\mathbb{R}} 

\newcommand{\nil}{\mathrm{nil}}

\AtBeginDocument{%
  \providecommand\BibTeX{{%
    \normalfont B\kern-0.5em{\scshape i\kern-0.25em b}\kern-0.8em\TeX}}}
\acmJournal{TCPS}

\begin{document}

\title{RobTL: A Temporal Logic for the Robustness of Cyber-Physical Systems}

\author{Valentina Castiglioni}
\orcid{0000-0002-8112-6523}
\affiliation{
  \institution{Reykjavik University}
  \city{Reykjavik}
  \country{Iceland}
}
\email{valentinac@ru.is}

\author{Michele Loreti}
\orcid{0000-0003-3061-863X}
\affiliation{
  \institution{University of Camerino}
  \city{Camerino}
  \country{Italy}
}
\email{michele.loreti@unicam.it}

\author{Simone Tini}
\orcid{0000-0002-3991-5123}
\affiliation{
  \institution{University of Insubria}
  \city{Como}
  \country{Italy}
}
\email{simone.tini@uninsubria.it}

\begin{abstract}
We propose the \emph{Robustness Temporal Logic} (\emph{RobTL}), a novel temporal logic for the specification and analysis of \emph{distances} between the behaviours of Cyber-Physical Systems (CPSs) over a finite time horizon.
Differently from classical temporal logic expressing properties on the behaviour of a system, we can use RobTL specifications to \emph{measure the differences} in the behaviours of systems with respect to various objectives and temporal constraints, and to \emph{study} how those differences \emph{evolve in time}.
Since the behaviour of CPSs is inevitably subject to uncertainties and approximations, we show how the unique features of RobTL allow us to specify property of \emph{robustness} of systems \emph{against perturbations}, i.e., their capability to function correctly even under the effect of perturbations.
Given the probabilistic nature of CPSs, our \emph{model checking algorithm} for RobTL specifications is based on \emph{statistical inference}. 
As an example of an application of our framework, we consider a supervised, self-coordinating engine system that is subject to attacks aimed at inflicting overstress of equipment.
\end{abstract}



\maketitle


\section{Introduction}
\label{sec:introduction}

Cyber-Physical Systems (CPSs)~\cite{RLSS10} are characterised by software applications, henceforth called \emph{programs}, that must be able to deal with \emph{highly changing operational conditions}, henceforth referred to as the \emph{environment}.
Examples of these applications are the software components of unmanned vehicles, controllers, medical devices, the devices in a smart house, etc.
In these contexts, the \emph{behaviour} of a system is the result of the \emph{interplay of programs with their environment}.

The main challenge in the analysis and verification of these systems is then the dynamical and, sometimes, unpredictable behaviour of the environment.
The highly changing behaviour of physical processes can only be approximated in order to become computationally tractable, and can thus constitute a safety hazard for the devices (e.g., a gust of wind for an unmanned aerial vehicle that is autonomously setting its trajectory); some devices may appear, disappear, or become temporarily unavailable; faults or conflicts may occur (e.g., the program responsible for the ventilation of a room may open a window, conflicting with the one that has to limit the noise level); sensors may introduce measurement errors; etc.
Moreover, there is a type of security threat which is unique to CPSs: cyber-physical attacks.
For instance, an attacker can induce a series of perturbations in the sensed data in order to entail some unexpected, hazardous, behaviour of the system.

It is therefore fundamental for these systems to be \emph{robust against uncertainties} (and perturbations), i.e., roughly speaking, to be able to function correctly even in their presence.
In the literature, we can find a wealth of proposals of robustness properties, that differ in the underlying model (and thus also in how uncertainties are modelled), in the mathematical formalisation, or in whether they are designed to analyse a specific feature of systems behaviour.
We refer the interested reader to \cite{FKP16,RT16,SF13,So08} for an overview of these notions.
Our purpose with this paper is not to introduce a new notion of robustness, nor to argue whether one proposal is better than an another one: it seems natural to us that different application contexts call for different formalisations of the notion of robustness.
Hence, our goal is to provide the tools for the specification of robustness properties of CPSs, i.e., of any \emph{measure} of the capability of a program to tolerate perturbations in the environmental conditions and still fulfil its tasks.
To this purpose, we need to \emph{compare} the behaviour of the system with its behaviour under the effect of perturbations.
More precisely, we need to \emph{measure the differences} between systems behaviours, possibly at \emph{different moments in time}.
Hence, whenever we require a system to be robust against perturbations, we are actually specifying a \emph{temporal property of distances between systems behaviours}.
To the best of our knowledge, no formal framework for the specification of similar properties has ever been proposed in the literature.
\begin{quote}
\emph{
Our principal target in this paper is then to introduce a formal framework allowing us to specify and analyse the properties of distances between the behaviours of systems operating in the presence of uncertainties.
}
\end{quote}
Our framework consists in:
\begin{itemize}
\item A \emph{model} for systems behaviour, i.e. the \emph{\traccione{}} from \cite{CLT21}.
\item A \emph{temporal logic} for the specification of the desired properties, i.e., the \emph{\logicName{}} (\emph{\logicShort}) that we introduce in this paper.
\item A \emph{model checking algorithm}, based on statistical inference, for the verification of \logicShort{} specifications, that we introduce in this paper and is available at {\small\color{blue}{\url{https://github.com/quasylab/jspear}}} as part of our \emph{Software Tool for the Analysis of Robustness in the unKnown environment} (\textsc{Stark}).
\end{itemize}


\paragraph{The model: \tracciones}

We adopt the discrete time model of \cite{CLT21} and represent the program-environment interplay in terms of the changes they induce on a set of application-relevant data,  called \emph{\dataspace}.
At each step, both the program and the environment induce some changes on the current state of the \dataspace, called \emph{\datastate}, providing thus a new \datastate{} at the next step.
Those modifications are also subject to the presence of uncertainties, meaning that it is not always possible to determine exactly the values assumed by data at the next step.
Hence, we model the changes induced at each step as a \emph{probability measure} on the attainable \datastates.
For instance, we can assume the computation steps of the system to be determined by a Markov kernel.
The behaviour of the system is then entirely expressed by its \emph{\traccione}, i.e., the sequence of probability measures over the \datastates{} obtained at each step.
In other words, the \traccione{} is the discrete-time version of the cylinder of \emph{all possible trajectories} of the system.


\paragraph{A novel temporal logic: \logicShort}

In the literature, quantitative extensions of model checking have been proposed, like stochastic (or probabilistic) model checking \cite{BdAFK18,Bai16,KNP07,KP12}, and statistical model checking \cite{SVA04,SVA05,ZPC13,HHA17,BMS16}.
These techniques rely either on a full specification of the system to be checked, or on the possibility of simulating the system by means of a Markovian model, or Bayesian inference on samples.
Then, quantitative model checking is based on a specification of requirements in a probabilistic temporal logic, such as PCTL \cite{HJ94}, CSL \cite{ASSB00,ASSB96}, probabilistic variants of LTL \cite{Pnu77}, etc.
Similarly, if Runtime Verification \cite{BFFR18} is preferred to off-line verification, probabilistic variants of MTL \cite{Koy90} and STL \cite{MN04} were proposed \cite{TH16,SK16}.
In temporal logics in the quantitative setting, uncertainties are usually dealt with by imposing probabilistic guarantees on a given property of the behaviour of a system to be satisfied.

However, these properties are usually evaluated over a single trajectory of the system, and they do not allow us to study the probabilistic transient behaviour of systems, which takes into account the combined effects of the program-environment interplay and uncertainties.
Besides, there are some properties, like the robustness properties outlined above, that can only be expressed in terms of requirements on the evolution of distances between systems behaviours.
We also remark that, in general, if we need to compare behaviours with respect to different targets in time, or to check for various properties that depend on different aspects of behaviour, we cannot use a single distance to obtain a meaningful analysis, but we need to combine and compare different measures (each capturing a specific feature of systems behaviour).

For all these reasons, we introduce the novel temporal logic \emph{\logicName{}} (\emph{\logicShort}) that not only allows us to analyse and compare distances between \tracciones{} over a finite time horizon, but it also provides the means to specify those distances.
Specifically, \logicShort{} offers:
\begin{itemize}
\item a class of \emph{expressions} for the arbitrary definition of distances between \tracciones, potentially taking into account different objectives of the system in time;
\item special atomic propositions for the evaluation and comparison of distances between a \traccione{} and its disruption via a perturbation;
\item (classical) Boolean and temporal operators for the analysis of the evolution of the specified distances over a finite time horizon.
\end{itemize}


\paragraph{Checking \logicShort{} specifications.}

We provide a \emph{statistical model checking algorithm} for the verification of \logicShort{} specifications, consisting of three components:
\begin{inparaenum}
\item A simulation procedure for the \tracciones{} of systems, and their perturbed versions.
\item A mechanism, based on statistical inference, for the evaluation of the distances over \tracciones{}.
\item A procedure that verifies whether a \logicShort{} formula is satisfied, by inspecting its syntax.
\end{inparaenum}
Since our algorithms are based on statistical inference, we need to take into account the statistical error in the evaluation of formulae.
Hence, we also propose a three-valued semantics for \logicShort{} specifications, in which the truth value \emph{unknown} is added to true and false.
The value \emph{unknown} suggests that the parameters used to specify the desired robustness property need some tuning, or (if these are fixed) that a larger number of samples is needed to obtain a more precise evaluation of the distances.

To strengthen our contribution, we apply our framework to the analysis of the \emph{security} of CPSs, which is one of the major open challenges (and, thus, hottest topics) in this research field \cite{TS21,BVMG12}.
In detail, we consider a case-study from Industrial Control Systems: an engine system that is subject to cyber-physical attacks aimed at inflicting overstress of equipment \cite{GGIKLW2015}.
Due to page limits, we present only a simplified version of the system, inspired by \cite{LMMT21}, and we show how \logicShort{} specifications can be used to measure the robustness of the engine against the attacks, which are implemented as perturbations over the \traccione{} of the system.

All the algorithms and examples (including the specifications of system, perturbations, and formulae) have been implemented in our tool \textsc{Stark}, available at {\small\color{blue}{\url{https://github.com/quasylab/jspear}}}.


\paragraph{Organisation of contents}

After a concise presentation of the mathematical background in Section~\ref{sec:background}, we give a bird's eye view on \tracciones{} in Section~\ref{sec:tracciones}.
In Section~\ref{sec:dist_pert} we introduce the two basic ingredients necessary for the definition of distances over \tracciones{} and robustness properties: a (hemi)metric over probability measures and a simple language to model perturbations on data.
The core of our paper is in Section~\ref{sec:requirements}, where we present the \logicName.
Then, in Section~\ref{sec:statisticalmc} we outline the statistical model checking algorithm for \logicShort specifications, the analysis of the statistical errors and the related three-valued semantics of \logicShort.
We conclude the paper by briefly discussing related and future work in Section~\ref{sec:conclusion}.


\section{Background}
\label{sec:background}

In this section we present the mathematical background on which we build our contribution.
We present in detail only the notions that are crucial to understand the development of our work.
Conversely, the notions that are needed simply to guarantee the mathematical correctness of the definitions in this section are not explained; the interested reader can find their formal definition in any Analysis textbook.


\paragraph{Measurable spaces and measurable functions}
A \emph{$\sigma$-algebra} over a set $\Omega$ is a family $\F$ of subsets of $\Omega$ s.t.:\ 
\begin{inparaenum}
\item
$\Omega \in \F$, 
\item $\F$ is closed under complementation;
and
\item
$\F$ is closed under countable union. 
\end{inparaenum}
The pair $(\Omega, \Sigma)$ is called a \emph{measurable space} and the sets in $\Sigma$ are called \emph{measurable sets}, ranged over by $\mesA,\mesB,\dots$.
For an arbitrary family $\Phi$ of subsets of $\Omega$, the $\sigma$-algebra \emph{generated} by $\Phi$ is the smallest $\sigma$-algebra over $\Omega$ containing $\Phi$.
In particular, we recall that given a topological space $\Omega$, the \emph{Borel $\sigma$-algebra} over $\Omega$, denoted $\borel(\Omega)$, is the $\sigma$-algebra generated by the open sets in the topology.
For instance, given $n\in \nats^+$, we can consider the $\sigma$-algebra $\borel(\reals^n)$ generated by the open intervals in $\reals^n$.
Given two measurable spaces $(\Omega_i,\Sigma_i)$, $i = 1,2$, the \emph{product $\sigma$-algebra} $\Sigma_1 \otimes \Sigma_2$ is the $\sigma$-algebra on $\Omega_1 \times \Omega_2$ generated by the sets $\{\mesA_1 \times \mesA_2 \mid \mesA_i \in \Sigma_i\}$.
Given measurable spaces $(\Omega_1,\Sigma_1), (\Omega_2,\Sigma_2)$,
a function $f \colon \Omega_1 \to \Omega_2$ is said to be $\Sigma_1$-\emph{measurable} if $f^{-1}\!(\mesA_2) \!\in\! \Sigma_1$ for all $\mesA_2 \!\in\! \Sigma_2$, with $f^{-1}(\mesA_2) \! = \! \{\omega \!\in\! \Omega_1 \mid f(\omega) \in \mesA_2\}$.


\paragraph{Probability spaces and random variables}

A \emph{probability measure} on a measurable space $(\Omega,\Sigma)$ is a function $\mu \colon \Sigma \to [0,1]$ s.t.\  
$\mu(\Omega) = 1$,
$\mu(\mesA) \ge 0$ for all $\mesA \in \Sigma$ and
$\mu( \bigcup_{i \in I} \mesA_i) = \sum_{i \in I}\mu(\mesA_i)$ for every countable family of pairwise disjoint measurable sets $\{\mesA_i\}_{i\in I} \subseteq \Sigma$.
Then $(\Omega,\Sigma, \mu)$ is called a \emph{probability space}.
We let $\distrib{(\Omega,\F)}$ denote the set of all probability measures over $(\Omega,\F)$.

For $\omega \in \Omega$, the \emph{Dirac measure} $\dirac_{\omega}$ is defined by $\dirac_\omega (\mesA) = 1$, if $\omega \in \mesA$, and $\dirac_{\omega}(\mesA) = 0$, otherwise, for all $\mesA \in \F$.
Given a countable set of reals $(p_i)_{i \in I}$ with $p_i \ge 0$ and $\sum_{i \in I}p_i = 1$, the \emph{convex combination} of the probability measures $\{\mu_i\}_{i \in I} \subseteq \distrib{(\Omega,\F)}$ is the probability measure $\sum_{i \in I} p_i \cdot \mu_i$ in $\distrib{(\Omega,\F)}$ defined by $(\sum_{i \in I} p_i \cdot \mu_i)(\mesA) = \sum_{i \in I} p_i \mu_i(\mesA)$, for all $\mesA \in \F$.
A probability measure $\mu \in \distrib{(\Omega,\F)}$ is called \emph{discrete} if $\mu=\sum_{i \in I}p_i \cdot \dirac_{\omega_i}$, with $\omega_i \in \Omega$, for some countable set of indexes $I$.

Assume a probability space $(\Omega, \F, \mu)$ and a measurable space $(\Omega', \F')$.
A function $X \colon \Omega \to \Omega'$ is called a \emph{random variable} if it is $\Sigma$-measurable.
The \emph{distribution measure}, or \emph{cumulative distribution function} (cdf), of $X$ is the probability measure $\mu_X$ on $(\Omega',\F')$ defined by $\mu_X(\mesA)=\mu(X^{-1}(\mesA))$ for all $\mesA \in \F'$.
Given random variables $X_i$ from $(\Omega_i,\F_i, \mu_i)$ to $(\Omega'_i,\F'_i)$, $i=1,\ldots,n$, the collection $\rv = [X_1,\ldots,X_n]$ is called a \emph{random vector}. 
The cdf of a random vector $\rv$ is given by the \emph{joint} distribution of the random variables in it. 

\begin{remark}
Since we will consider Borel sets over $\real^n$ ($n \ge 1$), in the examples and explanations throughout the paper, we will use directly the cdf of a random variable rather than formally introducing the probability measure defined on the domain space.
Similarly, when the cdf is absolutely continuous with respect to the Lebesgue measure, then we shall reason directly on the \emph{probability density function} (pdf) of the random variable (i.e., the Radon-Nikodym derivative of the cdf with respect to the Lebesgue measure). 
Consequently, we shall use the more suggestive, and general, term \emph{\bfseries distribution} in place of the terms probability measure, cdf and pdf.
\end{remark}


\paragraph{The Wasserstein hemimetric}

A \emph{metric} on a set $\Omega$ is a function $m \colon \Omega \times \Omega \to \real^{\ge0}$ s.t.\ $m(\omega_1,\omega_2) = 0$ if{f} $\omega_1 = \omega_2$, $m(\omega_1,\omega_2) = m(\omega_2,\omega_1)$, and $m(\omega_1,\omega_2) \le m(\omega_1,\omega_3) + m(\omega_3,\omega_2)$, for all $\omega_1,\omega_2,\omega_3 \in \Omega$.
We obtain a \emph{hemimetric} by relaxing the first property to $m(\omega_1,\omega_2) = 0$ if $\omega_1=\omega_2$, and by allowing $m$ to not be symmetric.
A (hemi-)metric $m$ is $l$-\emph{bounded} if $m(\omega_1,\omega_2) \le l$ for all $\omega_1,\omega_2 \in\Omega$.
For a (hemi-)metric on $\Omega$, the pair $(\Omega,m)$ is a (\emph{hemi-})\emph{metric space}.

Given a (hemi-)metric space $(\Omega,m)$, the (hemi-)metric $m$ induces a natural topology over $\Omega$, namely the topology generated by the open $\varepsilon$-balls, for $\varepsilon > 0$, $B_{m}(\omega,\varepsilon) = \{\omega' \in \Omega \mid m(\omega,\omega') < \varepsilon\}$.
We can then naturally obtain the Borel $\sigma$-algebra over $\Omega$ from this topology.

In this paper we will make use of \emph{hemimetrics on distributions}.
To this end we will make use of the Wasserstein lifting \cite{W69} whose definition is based on the following notions and results.
Given a set $\Omega$ and a topology $T$ on $\Omega$, the topological space $(\Omega,T)$ is said to be \emph{completely metrisable} if there exists at least one metric $m$ on $\Omega$ such that $(\Omega, m)$ is a complete metric space and $m$ induces the topology $T$.
A \emph{Polish space} is a separable completely metrisable topological space.
In particular, we recall that:\ 
\begin{inparaenum}[(i)]
\item $\real$ is a Polish space; and
\item every closed subset of a Polish space is in turn a Polish space.
\end{inparaenum}
Moreover, for any $n \in N$, if $\Omega_1, \dots, \Omega_n$ are Polish spaces, then the Borel $\sigma$-algebra on their product coincides with the product $\sigma$-algebra generated by their Borel $\sigma$-algebras, namely $\borel( \bigtimes_{i=1}^n \Omega_i) = \bigotimes_{i=1}^n \borel(\Omega_i)$ (see, e.g., \cite[Lemma 6.4.2]{Bo07}).
These properties of Polish spaces are interesting for us since they guarantee that all the distributions we consider in this paper are Radon measures and, thus, the Wasserstein lifting is well-defined on them.
For this reason, we also directly present the Wasserstein hemimetric by considering only distributions on Borel sets.

\begin{definition}
[Wasserstein hemimetric]
\label{def:Wasserstein}
Consider a Polish space $\Omega$ and let $m$ be a hemimetric on $\Omega$.
For any two distributions $\mu$ and $\nu$ on $(\Omega,\borel(\Omega))$, the \emph{Wasserstein lifting} of $m$ to a distance between $\mu$ and $\nu$ is defined by
\[
\Wasserstein(m)(\mu,\nu) = \inf_{\w \in \W(\mu,\nu)} \int_{\Omega \times \Omega} m(\omega,\omega') \dd\w(\omega,\omega')
\]
where $\W(\mu,\nu)$ is the set of the \emph{couplings of $\mu$ and $\nu$}, namely the set of joint distributions $\w$ over the product space $(\Omega \times \Omega, \borel(\Omega \times \Omega))$ having $\mu$ and $\nu$ as left and right marginal, respectively, namely $\w(\mesA \times \Omega) = \mu(\mesA)$ and $\w(\Omega \times \mesA) = \nu(\mesA)$, for all $\mesA \in \borel(\Omega)$.
\end{definition}

Despite the Wasserstein distance was originally defined on a metric on $\Omega$, the Wasserstein hemimetric given above is well-defined.
We refer the interested reader to \cite{FR18} and the references therein for a formal proof of this fact.
In particular, the Wasserstein hemimetric is given in \cite{FR18} as Definition 7 (considering the compound risk excess metric defined in Equation (31) of that paper), and Proposition 4 in \cite{FR18} guarantees that it is indeed a well-defined hemimetric on $\distrib(\Omega,\borel(\Omega))$.
Moreover, Proposition 6 in \cite{FR18} guarantees that the same result holds for the hemimetric $m(x,y) = \max\{y-x,0\}$ which will play an important role in our work (cf.\ Definition~\ref{def:metric_DS} below).

\begin{remark}
As elsewhere in the literature, for simplicity and brevity, we shall henceforth use the term \emph{metric} in place of the term hemimetric.
\end{remark}


\section{\Tracciones}
\label{sec:tracciones}

We focus on \emph{systems} consisting of a \emph{program} and an \emph{environment}, whose interaction produces changes on a \emph{shared data space}, containing the values assumed by \emph{physical quantities}, \emph{sensors}, \emph{actuators}, and the \emph{internal variables} of the program. 
Being the result of the system runs, the evolution of data must be captured by the semantic model.
This observation is behind the semantic model from \cite{CLT21}, which bases on \emph{\tracciones}, i.e., sequences of distributions over the values assumed by data over time.
In~\cite{CLT21}, CPSs were specified by means of simple programs having a discrete behaviour and reading/writing data at each time instant, and \emph{probabilistic evolution functions} expressing the effects of the environment on data between two time instants.
There, we gave a detailed, technical presentation of the generation of \tracciones{} from those specifications.
However, for our purposes in the present paper, it is not necessary to report the semantic mapping from \cite{CLT21} in full detail, since for the verification of properties in our logic we can abstract from the programs and environments that generated the \tracciones. 
Hence, in this section, we limit ourselves to recap a few basic ingredients that are necessary to obtain a well defined notion of \traccione.

Technically, a \emph{\dataspace{}} is defined by means of a \emph{finite} set of \emph{variables} $\Var$.
Without loss of generality, we assume that for each $x \in \Var$ the domain $\D_x \subseteq \real$ is either a \emph{finite set} or a \emph{compact} subset of $\real$.
Notice that this means that $\D_x$ is a Polish space.
Moreover, as a $\sigma$-algebra over $\D_x$ we assume the Borel $\sigma$-algebra, denoted $\borel_x$.
As $\Var$ is a finite set, we can always assume it to be ordered, namely $\Var=\{x_1,\dots,x_{n}\}$ for a suitable $n \in \nats$.

\begin{definition}
[\Dataspace]
We define the \emph{\dataspace{}} over $\Var$, notation $\D_{\Var}$, as the Cartesian product of the variables domains, namely $\D_{\Var} = \bigtimes_{i = 1}^n \D_{x_i}$.
Then, as a $\sigma$-algebra on $\D_\Var$ we consider the product $\sigma$-algebra $\borel_{\D_\Var} = \bigotimes_{i=1}^n \borel_{x_i}$.
\end{definition}

When no confusion arises, we use $\D$  for $\D_{\Var}$ and $\borel_\D$ for $\borel_{\D_{\Var}}$.
Elements in $\D$ are the $n$-ples of the form $(v_1,\dots,v_n)$, with $v_i \in \D_{x_i}$, that can be identified by means of functions $\ds \colon \Var \to \real$, with $\ds(x) \in \D_x$ for all $x \in \Var$. 
Each function $\ds$ identifies a particular configuration in the \dataspace, and it is thus called a \emph{\datastate}.

\begin{definition}
[\Datastate]
\label{def:datastate}
A \emph{\datastate{}} is a mapping $\ds \colon \Var \to \real $ from variables to values, with $\ds(x) \in \D_x$ for all $x \in \Var$. 
\end{definition}

\begin{notation}
In the examples that follow we will use a variable name $x$ to denote all: the variable $x$, the function describing the evolution in time of the values assumed by $x$, and the random variable describing the distribution of the values that can be assumed by $x$ at a given time.
The role of the name $x$ will always be clear from the context.
\end{notation}

\begin{figure}
\begin{subfigure}{0.39 \textwidth}
\scalebox{1.0}{
\begin{tikzpicture}
\draw[thick](1,3.4) rectangle (4,4.2);
\node at (2.5,3.8){\scalebox{0.85}{\textcolor{black}{$\mathsf{CONTROLLER}$}}};
\draw[thick,orange](1.0,4.2) rectangle (2.0,4.6);
\node at (1.5,4.4){\scalebox{1.0}{\textcolor{orange}{$\mathit{speed}$}}};
\draw[thick,orange](3.0,4.2) rectangle (4.0,4.6);
\node at (3.5,4.4){\scalebox{1.0}{\textcolor{orange}{$\mathit{cool}$}}};
\draw[thick,blue](4.0,3.6) rectangle (5.0,4.0);
\node at (4.5,3.8){\scalebox{1.0}{\textcolor{blue}{$\mathit{ch\_in}$}}};
\node at (5.15,3.8){\scalebox{0.85}{\textcolor{blue}{$\Leftarrow$}}};
\draw[thick](1,0) rectangle (4,1.2);
\node at (2.5,0.2){\scalebox{1.0}{\textcolor{black}{$\mathsf{IDS}$}}};
\node at (1.6,1.0){\scalebox{1.0}{\textcolor{purple}{$\mathit{stress}$}}};
\node at (1.9,0.6){\scalebox{1.0}{\textcolor{purple}{$\mathit{p}_1,\dots,\mathit{p}_6$}}};
\draw[thick,red](1.0,1.2) rectangle (2.0,1.6);
\node at (1.5,1.4){\scalebox{1.0}{\textcolor{red}{$\mathit{temp}$}}};
\draw[double,double distance=1.5,-{implies},blue,thick](1.5,1.6)--(1.5,2.4);
\draw[thick,blue](0.85,2.4) rectangle (2.15,2.8);
\node at (1.5,2.6){\scalebox{1.0}{\textcolor{blue}{$\mathit{ch\_temp}$}}};
\draw[double,double distance=1.5,-{implies},blue,thick](1.5,2.8)--(1.5,3.4);
\draw[thick,blue](4.0,0.0) rectangle (5.0,0.4);
\node at (4.5,0.2){\scalebox{1.0}{\textcolor{blue}{$\mathit{ch\_out}$}}};
\node at (5.15,0.2){\scalebox{1.0}{\textcolor{blue}{$\Rightarrow$}}};
\draw[thick,blue](4.0,0.8) rectangle (5.0,1.2);
\node at (4.5,1.0){\scalebox{1.0}{\textcolor{blue}{$\mathit{ch\_wrn}$}}};
\node at (5.15,1.0){\scalebox{1.0}{\textcolor{blue}{$\Rightarrow$}}};
\draw[thick,blue](2.8,1.8) rectangle (4.2,2.2);
\node at (3.5,2){\scalebox{1.0}{\textcolor{blue}{$\mathit{ch\_speed}$}}};
\draw[double,double distance=1.5,-{implies},blue,thick](3.5,1.2)--(3.5,1.8);
\draw[double,double distance=1.5,-{implies},blue,thick](3.5,2.2)--(3.5,3.4);
\end{tikzpicture}
}
\caption{Schema of the engine}
\label{subfig:schema}
\end{subfigure}
\begin{subfigure}{0.5 \textwidth}
\scriptsize
\begin{tabular}{c@{\;\;}c@{\;\;}c}
Name & Domain & Role \\
\hline
\textcolor{red}{$\mathit{temp}$} & \textcolor{red}{$[0,150]$} &
\textcolor{red}{sensor} detecting the temperature,
\\
& & accessed directly by IDS and
\\
& &
through \textcolor{blue}{$\mathit{ch\_temp}$} by CONTROLLER \\
\hline
\textcolor{orange}{$\mathit{speed}$} &
\textcolor{orange}{$\{\mathsf{slow}, \mathsf{half}, \mathsf{full}\}$} &
\textcolor{orange}{actuator} regulating the speed \\
\hline
\textcolor{orange}{$\mathit{cool}$} & \textcolor{orange}{$\{\mathsf{on}, \mathsf{off}\}$} &
\textcolor{orange}{actuator} regulating the cooling \\
\hline
\textcolor{blue}{$\mathit{ch\_temp}$} & \textcolor{blue}{$[0,150]$} &
insecure \textcolor{blue}{channel} \\
\hline
\textcolor{blue}{$\mathit{ch\_speed}$} & 
\textcolor{blue}{$\{\mathsf{slow}, \mathsf{half}\}$} &
\textcolor{blue}{channel} used by IDS to order to  
\\
& & controller to set the value of \textcolor{orange}{$\mathit{speed}$} \\
\hline
\textcolor{blue}{$\mathit{ch\_wrn}$} & \textcolor{blue}{$\{\mathsf{ok}, \mathsf{hot}\}$} &
\textcolor{blue}{channel} used to raise warnings \\
& & 
in case of anomalies \\
\hline
\textcolor{blue}{$\mathit{ch\_out}$} & 
\textcolor{blue}{$\{\mathsf{half}, \mathsf{full}\}$} &
\textcolor{blue}{channel} used to send requests \\
& & to other engines
\\
\hline
\textcolor{blue}{$\mathit{ch\_in}$} & 
\textcolor{blue}{$\{\mathsf{half}, \mathsf{full}\}$} 
& \textcolor{blue}{channel} dual to \textcolor{blue}{$\mathit{ch\_out}$} \\
\hline
\textcolor{purple}{$\mathit{p_1},..,\mathit{p_6}$}  & \textcolor{purple}{$[0,150]$} 
&
\textcolor{purple}{internal variables} storing the last \\
& &  6 temperatures detected by \textcolor{red}{$\mathit{temp}$} \\ 
\hline
\textcolor{purple}{$\mathit{stress}$} & 
\textcolor{purple}{$[0,1]$} &
\textcolor{purple}{internal var.} carrying stress level \\
\end{tabular}
\caption{The variables}
\label{subfig:variables}
\end{subfigure}

\begin{subfigure}{1.0 \textwidth}
\scriptsize
\[
\begin{array}{l@{\;}l@{\;}l}
\mathsf{Eng}
& 
=
&
\mathsf{Ctrl} \parallel \mathsf{IDS} 
\text{ \textcolor{ForestGreen}{ //  The symbol ``$\parallel$'' denotes the classical parallel composition operator over processes}}
\\
\mathsf{Ctrl}
& 
=
&
\mathrm{if}\; 
[\textcolor{blue}{\mathit{ch\_temp}} \ge 99.8]\; 
\wact{\mathsf{on}}{\textcolor{orange}{\mathit{cool}}}.
\mathsf{Cooling}
\;
\mathrm{else}
\; 
\mathsf{Check} 
\\
& & 
\text{ \textcolor{ForestGreen}{// If temperature is too high then cooling is activated by setting actuator $\textcolor{orange}{\mathit{cool}}$ to $\mathsf{on}$}
(by assignment $\wact{\mathsf{on}}{\textcolor{orange}{\mathit{cool}}}$).}
\\
& &
\text{ \textcolor{ForestGreen}{// In this case, since prefixing ``$.$'' consumes one unit of time, $\mathsf{Cooling}$  will start at the next instant}}
\\
\mathsf{Cooling}
& 
=
&
\surd . \surd . \surd . \surd . 
\mathsf{Check}
\text{ \textcolor{ForestGreen}{//  the cooling is kept on for 4 more instants ($\surd$ consumes one time unit without affecting variables) }}
\\
\mathsf{Check}
& 
=
&
\mathrm{if}\; 
[\textcolor{blue}{\mathit{ch\_speed}} = \mathsf{slow}]\;
(\wact{\mathsf{slow}}{\textcolor{orange}{\mathit{speed}}},
\wact{\mathsf{off}}{\textcolor{orange}{\mathit{cool}}}).
\mathsf{Ctrl} \;
\\
& &
\text{\textcolor{ForestGreen}{ // If a slow down order comes from  IDS through channel \textcolor{blue}{$\mathit{ch\_speed}$}, then actuator \textcolor{orange}{$\mathit{speed}$} is set to $\mathsf{slow}$}}
\\ & &
\mathrm{else}\;
(\wact{\textcolor{blue}{\mathit{ch\_in}}}{\textcolor{orange}{\mathit{speed}}},
\wact{\mathsf{off}}{\textcolor{orange}{\mathit{cool}}}).
\mathsf{Ctrl} \;
\\
& &
\text{\textcolor{ForestGreen}{ // Otherwise, any speed regulation request from other engines, received through channel \textcolor{blue}{$\mathit{ch\_in}$}, is satisfied}}
\\
\mathsf{IDS}
& 
=
&
\mathrm{if}\; 
[\textcolor{red}{\mathit{temp}} > \mathsf{101} \wedge \textcolor{orange}{\mathsf{cool}} = \mathsf{off}]\;
\text{\textcolor{ForestGreen}{ // If this guard is true then there is an anomaly }}
\\ & &
(\wact{\mathsf{hot}}{\textcolor{blue}{\mathit{ch\_wrn}}},
\wact{\mathsf{slow}}{\textcolor{blue}{\mathit{ch\_speed}}},
\wact{\mathsf{full}}{\textcolor{blue}{\mathit{ch\_out}}}).
\mathsf{IDS}
\\
& & 
\text{\textcolor{ForestGreen}{ // If there is an anomaly, a warning is raised on channel \textcolor{blue}{$\mathit{ch\_wrn}$}, a slow down order is sent to $\mathsf{Ctrl}$ through channel}}
\\
& &
\text{\textcolor{ForestGreen}{//  \textcolor{blue}{$\mathit{ch\_speed}$}, and a speed up request is sent to
other engines through channel \textcolor{blue}{$\mathit{ch\_out}$}}}
\\ & &
\mathrm{else}\;
(\wact{\mathsf{ok}}{\textcolor{blue}{\mathit{ch\_wrn}}},
\wact{\mathsf{half}}{\textcolor{blue}{\mathit{ch\_speed}}},
\wact{\mathsf{half}}{\textcolor{blue}{\mathit{ch\_out}}}).
\mathsf{IDS}\\
& &
\text{ \textcolor{ForestGreen}{ //  Otherwise, the order and request are to proceed at regular speed}}
\end{array}
\]
\caption{The program $\mathsf{Eng}$
}
\label{subfig:logic}
\end{subfigure}

\begin{subfigure}{1.0\textwidth}
\scriptsize
\begin{equation*}
\label{eq:enginedynamics}
\begin{array}{r@{\;}l}
\textcolor{purple}{\mathit{p_k}}(\tau + 1) 
\; = & 
\begin{cases}
\begin{array}{cl}
\textcolor{red}{\mathit{temp}}(\tau) 
&
\text{ if } k=1
\text{ \textcolor{ForestGreen}{// internal variables $\textcolor{purple}{p_1},\dots,\textcolor{purple}{p_6}$ store the last 6 temperatures}}
\\
\textcolor{purple}{\mathit{p_{k-1}}}(\tau) &
\text{ if }
k=2,\dots,6
\end{array}
\end{cases}
\\[3.0 ex]
\textcolor{purple}{\mathit{stress}}(\tau + 1) 
\; = & 
\begin{cases}
\begin{array}{cl@{}l}
\max(1,\textcolor{purple}{\mathit{stress}}(\tau) + \mathsf{stressincr})
&
\text{if }
|\{k \mid \mathit{p_k}(\tau) \ge 100\} | > 3
&
\text{ \textcolor{ForestGreen}{ // \textcolor{purple}{$\mathit{stress}$} is augmented by constant $\mathsf{stressincr}$ iff}}
\\
\textcolor{purple}{\mathit{stress}}(\tau) 
&
\text{otherwise}
&
\text{ \textcolor{ForestGreen}{ // temp was too high for $>3$ instants over 6 }}
\end{array}
\end{cases}
\\[3.0 ex]
\textcolor{red}{\mathit{temp}}(\tau + 1) 
\; = & 
\textcolor{red}{\mathit{temp}}(\tau ) + v
\text{ \textcolor{ForestGreen}{ // value detected by 
\textcolor{red}{$\mathit{temp}$} varies by a value $v$ that is uniformly distributed in an interval, whose}}
\\
\textcolor{blue}{\mathit{ch\_temp}}(\tau + 1) 
\; = &
\textcolor{red}{\mathit{temp}}(\tau ) + v
\text{ \textcolor{ForestGreen}{ // endpoints depend on the state of actuators. The same value is taken by \textcolor{blue}{$ch\_temp$} in no attack case}}
\\[3.0 ex]
v \; \sim &
\begin{cases}
\begin{array}{cl}
{\mathcal U}[-1.2,-0.8] 
& 
\text{ if } 
\textcolor{orange}{\mathit{cool}}(\tau) = \mathsf{on}
\text{
\textcolor{ForestGreen}{ //
$v$ is negative if cooling is activated
 }}
\\
{\mathcal U}[0.1,0.3] 
& 
\text{ if } 
\textcolor{orange}{\mathit{cool}}(\tau) = \mathsf{off}
\text{ and }
\textcolor{orange}{\mathit{speed}}(\tau) = \mathsf{slow}
\text{ \textcolor{ForestGreen}{  //
otherwise $v$ is positive and depends on the speed
 }}
\\
{\mathcal U}[0.3,0.7] 
& 
\text{ if } 
\textcolor{orange}{\mathit{cool}}(\tau) = \mathsf{off}
\text{ and }
\textcolor{orange}{\mathit{speed}}(\tau) = \mathsf{half}
\\
{\mathcal U}[0.7,1.2] 
& 
\text{ if } 
\textcolor{orange}{\mathit{cool}}(\tau) = \mathsf{off}
\text{ and }
\textcolor{orange}{\mathit{speed}}(\tau) = \mathsf{full}.
\end{array}
\end{cases}
\end{array}
\end{equation*}
\caption{The environment}
\label{subfig:environment}
\end{subfigure}
\caption{The engine system}
\label{fig:engine}
\end{figure}

\begin{example}
\label{ex:engine_1}
We exemplify all our notions on a case study from~\cite{LMMT21}, sketched in Figure~\ref{fig:engine}, consisting in a refrigerated engine system. 
The program of this CPS has three tasks:
\begin{inparaenum}[(i)]
\item
regulate the engine speed, 
\item 
maintain the temperature within a specific range by means of a cooling system, and 
\item 
detect anomalies.
\end{inparaenum}
The first two tasks are on charge of a controller, the last one is took over by an intrusion detection system, henceforth IDS.
As shown in Figure~\ref{subfig:schema}, these two components use channels (depicted in \textcolor{blue}{blue}) to exchange information and to communicate with other programs.

The variables used in the system, and their role, are listed in Figure~\ref{subfig:variables}.
Notice that \textcolor{purple}{$\mathit{stress}$} is a variable that quantifies the level of equipment stress, which increases when the temperature stays too often above the threshold $100$, the idea being that the higher the stress the higher the probability of a wreckage.
The programs and the environment acting on these data are in Figures~\ref{subfig:logic} and~\ref{subfig:environment}, respectively.
At each scan cycle the program sets internal variables (depicted in \textcolor{purple}{purple}) and actuators (depicted in \textcolor{orange}{orange}) according to the values received from sensors (depicted in \textcolor{red}{red}), the IDS raises a warning if the status of sensors and actuators is unexpected, and the environment models the probabilistic evolution of the temperature.
We remark that the controller 
must use channel \textcolor{blue}{$\mathit{ch\_temp}$} to receive data from sensor \textcolor{red}{$\mathit{temp}$}.
Even though the use of channels is a common feature in CPSs, it exposes them to attacks, like those aimed at inflicting overstress of equipment \cite{GGIKLW2015} that we will present in Example~\ref{ex:perturbation}.
We assume that the engine can cooperate with other engines (e.g., in an aircraft with a \emph{left} and a \emph{right} engine), by receiving values on channel \textcolor{blue}{$\mathit{ch\_in}$} and sending values on \textcolor{blue}{$\mathit{ch\_out}$}.
If the engine cannot work at regular ($\mathsf{half}$) speed and must work at $\mathsf{slow}$, it asks to other engines to proceed at $\mathsf{full}$ speed to compensate the lack of performance.
\end{example}

In~\cite{CLT21}, an \emph{\traccione{}} is defined as a sequence of  distributions over \datastates{}.
This sequence is countable as a discrete time approach is adopted.
In the present paper we do not focus on how evolution sequences are generated: we simply assume a function $\cstep \colon \D \to \distrib(\D,\borel_\D)$ governing the evolution of the system, that is a \emph{Markov kernel}, and our evolution sequence is the \emph{Markov process} generated by $\cstep$ (see Definition~\ref{def:traccione} below).
Specifically, $\cstep(\ds)(\mesD)$ expresses the probability to reach a \datastate{} in $\mesD$ from the \datastate{} $\ds$ in one computation step.
Clearly, each system is characterised by a particular function $\cstep$.
For instance, for the engine system in Example~\ref{ex:engine_1}, function $\cstep$ is obtained following~\cite{CLT21} by combining the effects of the program in Figure~\ref{subfig:logic} and of the environment in Figure~\ref{subfig:environment}.
Moreover, each system will start its computation from a determined distribution over \datastates, called the \emph{initial distribution}.

\begin{definition}
[\Traccione]
\label{def:traccione}
Assume a Markov kernel $\cstep \colon \D \to \distrib(\D,\borel_\D)$ generating the behaviour of a system $\system$ having $\mu$ as initial distribution.
Then, the \emph{\traccione{}} of $\system$ is a countable sequence $\ES_\mu = \ES_\mu^0,\ES_\mu^1,\dots$ of distributions in $\distrib(\D,\borel_\D)$ such that, for all $\mesD \in \borel_{\D}$: 
\[
\ES_{\mu}^{0}(\mesD)  = \mu(\mesD)
\qquad\qquad
\ES_{\mu}^{i+1}(\mesD)  =  \int_{\D} \cstep(\ds)(\mesD) \dd \ES_{\mu}^{i}(\ds).
\]
\end{definition}

\begin{notation}
\label{notation:no_mu}
In order to avoid an overload of notation and possible confusion with symbol $\ds$, in the integrals we use symbol $\dd$ to denote the differential in place of the classical $d$.

Moreover, we shall write $\ES, \ES_1$ in place of, respectively, $\ES_\mu,\ES_{\mu_1}$, whenever the form of the initial distributions $\mu$ and $\mu_1$ do not play a direct role in the result, or definition, that we are discussing.
\end{notation}


\section{Towards distances between evolution sequences}
\label{sec:dist_pert}

Our principal target is to provide a temporal logic allowing us to study various robustness properties of systems, which boils down to being able to measure both, the effect of perturbations on the behaviour of a given system, and the capability of a (perturbed) system to fulfil its original tasks.
Hence, before proceeding to discuss the logic (in Section~\ref{sec:requirements}), we devote this section to a presentation of a means to establish how well a system is fulfilling its tasks, and a formalisation of perturbations and their effects on \tracciones.

Informally, the idea is to introduce a \emph{distance over distributions on \datastates{}} measuring their differences with respect to a \emph{given target}, and then use the operators of the logic to extend it to the \tracciones, while possibly taking into consideration \emph{different objectives and perturbations in time}.
In our setting, as in most application contexts, we have that both, the tasks of the system and possible perturbations over its behaviour, can be expressed in a purely data-driven fashion.
In the case of tasks, at any time step, any difference between the desired value of some parameters of interest and the data actually obtained can be interpreted as a flaw in the behaviour of the system.
Similarly, we can imagine perturbations as some random noise introduced on data, and represent them as functions mapping each datum in a distribution over data.
This noise is then propagated, in time, along the \traccione.


\subsection{A distance between distributions over data states}
\label{sec:metriche}

Following \cite{CLT21}, to capture how well a system is fulfilling a particular task, we use a \emph{penalty function} $\rho \colon \D \to [0,1]$, i.e., a continuous function that assigns to each \datastate{} $\ds$ a penalty in $[0,1]$ expressing how far the values of the parameters of interest (i.e., related to the considered task) in $\ds$ are from their desired ones.
Hence, $\rho(\ds) = 0$ if $\ds$ respects all the parameters. 
We can then use a penalty function $\rho$ to obtain a \emph{distance on \datastates}, i.e., the $1$-bounded \emph{hemimetric} $m_{\rho}$:
$m_{\rho}(\ds_1,\ds_2)$ expresses how much $\ds_2$ is \emph{worse} than $\ds_1$ according to the parameters of interest.
Since some parameters can be time-dependent, so is $\rho$: the $\tau$-penalty function $\rho_\tau$ is evaluated on the \datastates{} with respect to the values of the parameters expected at time $\tau$.

\begin{definition}
[Metric on \datastates{}]
\label{def:metric_DS}
For any time step $\tau$, let $\rho_{\tau} \colon \D \rightarrow [0,1]$ be the $\tau$-penalty function on $\D$.
The $\tau$-\emph{metric on \datastates{}} in $\D$, $m_{\rho,\tau} \colon \D \times \D \to [0,1]$, is defined, for all $\ds_1,\ds_2 \in \D$, by $m_{\rho,\tau}(\ds_1,\ds_2) = \max\{\rho_{\tau}(\ds_2)-\rho_{\tau}(\ds_1), 0 \}$.
\end{definition}

Notice that $m_{\rho,\tau}(\ds_1,\ds_2) > 0$ if and only if $\rho_\tau(\ds_2) > \rho_\tau(\ds_1)$, i.e., the penalty assigned to $\ds_2$ is higher than that assigned to $\ds_1$.
For this reason, we say that $m_{\rho,\tau}(\ds_1,\ds_2)$ expresses how much $\ds_2$ is worse than $\ds_1$ with respect to $\rho_\tau$.

\begin{example}
\label{ex:penalty}
Consider the engine system from Example~\ref{ex:engine_1}.
The penalty functions $\rho^{\mathit{wrn}}$, $\rho^{\mathit{temp}}$, and $\rho^{\mathit{stress}}$, defined, for all time steps $\tau$, by:
\[
\rho^{\mathit{wrn}}_\tau(\ds) = \begin{cases}
1 & \text{if }\ds(ch\_wrn) = \mathsf{hot},
\qquad\qquad 
\!\!\rho^{\mathit{temp}}_\tau(\ds) = \sfrac{|\ds(\mathit{ch\_temp})-\ds(\mathit{temp})|}{150}
\\
0 & \text{if }\ds(ch\_wrn) = \mathsf{ok}\;\,
\qquad\qquad
\rho^{\mathit{stress}}_\tau(\ds) = \ds(stress)
\end{cases}
\]
express, respectively, how far the level of alert raised by the IDS, the value carried by channel $\mathit{ch\_temp}$, and the level of stress  
are from their desired value.
These coincide with the value $\mathsf{ok}$, the value of sensor $\mathit{temp}$, and zero, respectively.
Penalty functions may be used also to express \emph{false negatives} and \emph{false positives}, which represent the average \emph{effectiveness}, and the average \emph{precision}, respectively, of the IDS to signal through channel $\mathit{ch\_wrn}$ that the engine system is under stress.
In detail, we may add two variables, $\mathit{fn}$ and $\mathit{fp}$, that quantify false negatives and false positives depending on $\mathit{stress}$ and $\mathit{ch\_wrn}$.
Both variables are initialised to $0$ and updated by adding 
\[
\scalebox{0.87}{
$\mathit{fn}(\tau+1) =
\frac{\tau * \mathit{fn}(\tau) + \max(0,\mathit{stress(\tau)}-\mathit{ch\_wrn(\tau)})}
{\tau+1}
\quad
\mathit{fp}(\tau+1) =
\frac{\tau * \mathit{fp}(\tau)  + 
\max(0,\mathit{ch\_wrn(\tau)}-\mathit{stress}(\tau))}{\tau+1}
$
}
\]
to Figure~\ref{subfig:environment}, with values $\mathsf{ok}$ and $\mathsf{hot}$ of $\mathit{ch\_wrn}$ interpreted as $0$ and $1$, respectively.
Then, false negatives and false positives are expressed by penalty functions $\rho^{\mathit{fn}}_\tau(\ds) = \ds(\mathit{fn})$ and  $\rho^{\mathit{fp}}_\tau(\ds) = \ds(\mathit{fp})$, respectively, expressing how far false negatives and false positives are from their ideal value $0$.
\end{example}

We now need to lift the hemimetric $m_{\rho,\tau}$ to a hemimetric over $\distrib(\D,\borel_\D)$.
In the literature, we can find a wealth of notions of function liftings doing so (see \cite{RKSF13} for a survey).
Among those, the \emph{Wasserstein lifting} \cite{W69} (introduced in Definition~\ref{def:Wasserstein} in Section~\ref{sec:background}) has been applied in several different contexts, from optimal transport \cite{Vil08} to process algebras, for the definition of \emph{behavioural metrics} (e.g., \cite{DGJP99,CLT20}), and from privacy \cite{CCP18,CCP20} to machine learning (e.g., \cite{ACB17,TBGS18}) for improving the stability of generative adversarial networks training.
We opted for this lifting since:
\begin{inparaenum}[(i)]
\item it preserves the properties of the ground metric, and
\item we can apply statistical inference to obtain good approximations of it, whose exact computation is tractable \cite{TK09,SFGSL12,CLT21arxiv}.
\end{inparaenum}

\begin{definition}
[Distance over $\distrib(\D,\borel_\D)$]
\label{def:metric_prob_DS}
We lift the metric $m_{\rho,\tau}$ to a metric over $\distrib(\D,\borel_\D)$ by means of the Wasserstein lifting as follows: for any two distributions $\mu,\nu$ on $(\D,\borel_\D)$
\[
\Wasserstein(m_{\rho,\tau})(\mu,\nu) = \inf_{\w \in \W(\mu,\nu)} \int_{\D \times \D} m_{\rho,\tau}(\ds,\ds') \dd\w(\ds,\ds').
\]
\end{definition}


\subsection{Perturbations}
\label{sec:perturbations}

We now proceed to formalise the notion of perturbations and their effects on \tracciones.
Intuitively, a \emph{perturbation} is the effect of unpredictable events on the current state of the system.
Hence, we find it natural to model it as a function that maps a \datastate{} into a distribution over \datastates.
To account for possibly repeated, or different effects in time of a single perturbation, we make the definition of perturbation function also time-dependent.
Informally, we can think of a perturbation function $\pfun$ as a list of mappings in which the $i$-th element describes the effects of $\pfun$ at time $i$.
As we are in a the discrete-time setting, we identify time steps with natural numbers.

\begin{definition}
[Perturbation function]
\label{def:perturbation}
A \emph{perturbation function} is a mapping $\pfun \colon \D \times \nat \to \distrib(\D,\borel_\D)$ such that, for each $\tau \in \nat$, $\pfun(\cdot,\tau) \colon \D \to \distrib(\D,\borel_\D)$ is such that the mapping $\ds \mapsto \pfun(\ds,\tau)(\mesD)$ is $\borel_\D$-measurable
for all $\mesD \in \borel_\D$. 
\end{definition}

We remark that to model the fact that the perturbation function $\pfun$ has no effects on the system at time $\tau$, it is enough to define $\pfun(\ds,\tau) = \dirac_\ds$ for all $\ds \in \D$.

To describe the perturbed behaviour of a system we need to take into account the effects of a given perturbation function $\pfun$ on its \traccione.
This can be done by combining $\pfun$ with the Markov kernel $\cstep$ describing the evolution of the system.

\begin{definition}
[Perturbation of an \traccione]
\label{def:traccione_perturbato}
Given an \traccione{} $\ES_\mu$, with $\mu$ as initial distribution, and a perturbation function $\pfun$, we define the \emph{perturbation of $\ES_\mu$ via $\pfun$} as the \traccione{} $\ES^\pfun_{\mu}$ obtained as follows: 
\begin{align*}
\ES^{\pfun,0}_{\mu} (\mesD) ={} & 
\int_{\D} \pfun(\ds,0)(\mesD) \dd \mu(\ds)
\\
\ES^{\pfun,i+1}_{\mu}(\mesD) ={} &
\int_{\D} \left(
\int_{\D} \pfun(\ds',i+1)(\mesD) \dd \cstep(\ds)(\ds')
\right)
\dd \ES^{\pfun,i}_{\mu}(\ds),
\end{align*}
where function $\cstep$ is the Markov kernel that generates $\ES_\mu$.
\end{definition}


\paragraph{Specifying perturbations}

We specify a perturbation function $\pfun$ via a (syntactic) perturbation $\p$ in following language $\Perturbations$:
\[
\p \;::= \quad
\nil \quad \mid \quad
\f@\tau \quad \mid \quad
\p_1 \,;\, \p_2 \quad \mid \quad
\p^n
\]
where $\p$ ranges over $\Perturbations$, $n$ and $\tau$ are finite natural numbers, and:
\begin{itemize}
\item $\nil$ is the perturbation with \emph{no effects}, i.e., at each time step it behaves like the identity function $\mathtt{id} \colon \D \to \distrib(\D,\borel_\D)$ such that $\mathtt{id}(\ds) = \delta_\ds$ for all $\ds \in \D$;
\item $\f@\tau$ is an \emph{atomic perturbation}, i.e., a function $\f \colon \D \to \distrib(\D,\borel_\D)$ such that the mapping $\ds \mapsto \f(\ds)(\mesD)$ is $\borel_\D$-measurable for all $\mesD \in \borel_\D$, and that is applied precisely after $\tau$ time steps from the current instant;
\item $\p_1 \,;\, \p_2$ is a \emph{sequential perturbation}, i.e., perturbation $\p_2$ is applied at the time step subsequent to the (final) application of $\p_1$;
\item $\p^n$ is an \emph{iterated perturbation}, i.e., perturbation $\p$ is applied for a total of $n$ times.
\end{itemize}
Despite its simplicity, this language allows us to define some non-trivial perturbation functions that we can use to test systems behaviour.
The following example serves as a demonstration in the case of the engine system.
(Unspecified perturbations are assumed to behave like the identity.)

\begin{example}
\label{ex:perturbation}
In~\cite{LMMT21} several cyber-physical attacks tampering with sensors or actuators of the engine system aiming to inflict \emph{overstress of equipment}~\cite{GGIKLW2015} were described.
Here we show how those attacks can be modelled by employing our perturbations.
Consider sensor $\mathit{temp}$.
There is an attack that tricks the controller by adding a negative offset 
$o \in \real^{\le 0}$, uniformly distributed in an interval $[l_o,r_o]$,
to the value carried by the insecure channel $\mathit{ch\_temp}$ for $n$ units of time.
This attack aims to delay the cooling phase, forcing the system to work for several instants at high temperatures thus accumulating stress.
It can be modelled via the perturbation $\p_{\mathit{temp},o,\tau,n} = (\mathtt{id}@0)^{\tau};(\f_{\mathit{temp},o}@0)^n$ where $\tau$ is the time at which the attack starts, and $\f_{\mathit{temp},o}(\ds)$ is the distribution of the random variable $O(o,\ds)$, for $o \sim \mathcal{U}[l_o,r_o]$, defined for all $\ds \in \D$ by $O(o,\ds) = \ds'$ where $\ds'(\mathit{ch\_temp}) = \ds(\mathit{ch\_temp}) + o$, and $\ds'(x) = \ds(x)$ for all other variables in Figure~\ref{subfig:variables}.

Consider now actuator $\mathit{cool}$.
There is an attack that interrupts the cooling phase by switching off the cooling system as soon as the temperatures goes below $99.8 - t$ degrees, for some positive value $t \in \real^{\ge 0}$.
This attack, aiming to force the system to reach quickly high temperatures after the beginning of a cooling phase, is \emph{stealth}, meaning that the IDS does not detect it. 
It is modelled by using the perturbation $\p_{\mathit{cool},t,n}$ defined by
$\p_{\mathit{cool},t,n} = (\f_{\mathit{cool},t}@0)^n$, where 
$\f_{\mathit{cool},t}(\ds) = \dirac_\ds$, if $\ds(\mathit{temp}) \ge 99.8-t$, and 
$\f_{\mathit{cool},t}(\ds) = \delta_{\ds''}$, if $\ds(\mathit{temp}) < 99.8-t$, where $\ds''(\mathit{cool}) = \mathsf{off}$ and
$\ds''(x)=\ds(x)$ for all other variables in Figure~\ref{subfig:variables}.
\end{example}

Each $\p \in \Perturbations$ denotes a perturbation function as in Definition~\ref{def:perturbation}, namely a mapping of type $\D \times \nat \to \distrib(\D,\borel_\D)$.
To obtain it, we make use of two auxiliary functions: $\mathsf{effect}(\p)$, that describes the effect of $\p$ at the current step, and $\mathsf{next}(\p)$, that identifies the perturbation that will be applied at the next step.
Both functions are defined inductively on the structure of perturbations.
\begin{align*}
&
\mathsf{effect}(\nil) = \mathtt{id}
& &
\mathsf{effect}(\f@\tau) = 
\begin{cases}
\mathtt{id} & \text{ if } \tau>0, \\
\f & \text{ if } \tau=0
\end{cases}
\\
&
\mathsf{effect}(\p^n) = \mathsf{effect}(\p)
& &
\mathsf{effect}(\p_1; \p_2) = \mathsf{effect}(\p_1)
\\
&
\mathsf{next}(\nil) = \nil
& &
\mathsf{next}(\f@\tau) = 
\begin{cases}
\f@(\tau-1) & \text{ if } \tau>0, \\
\nil  & \text{ otherwise}
\end{cases}
\\
&
\mathsf{next}(\p^n) =
\begin{cases}
\mathsf{next}(\p);\p^{n-1} & \text{ if } n>0,\\
\nil  & \text{ otherwise}
\end{cases}
& &
\mathsf{next}(\p_1; \p_2) = 
\begin{cases}
\mathsf{next}(\p_1); \p_2 & \text{ if } \mathsf{next}(\p_1) \neq \nil, \\
\p_2 & \text{ otherwise.}
\end{cases}
\end{align*}
We can now define the semantics of perturbations as the mapping $\psem{\cdot} \colon \Perturbations \to ( \D \times \nat \to \distrib(\D,\borel_\D))$ such that, for all $\ds \in \D$ and $i \in \nat$:
\[
\psem{\p}(\ds,i) = \mathsf{effect}(\mathsf{next}^i(\p))(\ds),
\]
where $\mathsf{next}^0(\p) = \p$ and $\mathsf{next}^i(\p) = \mathsf{next}(\mathsf{next}^{i-1}(\p))$, for all $i > 0$.

\begin{proposition}
For each $\p \in \Perturbations$, the mapping $\psem{\p}$ is a perturbation function.
\end{proposition}

\begin{proof}
Since, by definition, each $\f$ occurring in atomic perturbations is such that $\ds \mapsto \f(\ds)(\mesD)$ is $\borel_\D$-measurable for all $\mesD \in \borel_\D$, and the same property trivially holds for the identity function $\mathtt{id}$, it is immediate to conclude that $\psem{\p}$ satisfies Definition~\ref{def:perturbation} for each $\p \in \Perturbations$.
\end{proof}


\section{The \logicName}
\label{sec:requirements}

We now present our \emph{\logicName{}} (\emph{\logicShort}), which is the core of our tool for the specification and analysis of distances between nominal and perturbed \tracciones{} over a finite time horizon, denoted by $\h$. This allows for studying systems robustness against perturbations.

This is made possible by combining novel atomic propositions, allowing for the evaluation of various distances between a distribution in the \traccione{} of the nominal system and a distribution in the \traccione{} of the perturbed system, with classical Boolean and temporal operators allowing for extending these evaluations to the entire \tracciones{}.
Informally, we use the atomic proposition $\Delta(\esp,\p) \bowtie \eta$ to evaluate the distance specified by an expression $\exp$ at a given time step between a given \traccione{} and its perturbation, specified by some $\p \in \Perturbations$, and to compare it with the threshold $\eta$.

\begin{definition}[\logicShort]
\label{def:evtl}
The modal logic \logicShort{} consists in the set of formulae $\logicSymbol$ defined by:
\[
\varphi\; :: =
\top \;\; \mid \;\;
\Delta(\esp,\p) \bowtie \eta \;\; \mid \;\;
\neg \varphi \;\; \mid \;\;
\varphi \wedge \varphi \;\; \mid \;\;
\funtil{\varphi}{I}{\varphi}
\]
where $\varphi$ ranges over $\logicSymbol$, ${\bowtie} \in \{<,\le,\ge,>\}$, $\eta \in [0,1]$, $\p$ is a perturbation function, $I \subseteq [0,\h]$ is a bounded time interval, and $\esp$ ranges over expressions in $\espClass$ and defined syntactically as follows:
\begin{align*}
\esp\; :: ={} &
\sx{\rho} \;\; \mid \;\;
\dx{\rho} \;\; \mid \;\;
\eventually{I} \esp \;\;\mid\;\;
\always{I} \esp \;\;\mid\;\;
\esp \until{I} \esp \;\;\mid
\\
&
\meet (\esp, \esp) \;\;\mid\;\;
\join (\esp, \esp) \;\;\mid\;\;
\sum_{k \in K} w_k \cdot \esp_k \;\;\mid\;\;
\sigma(\esp,\bowtie \zeta) 
\end{align*}
where $\rho$ ranges over penalty functions, $K$ is a finite set of indexes, $w_k \in (0,1]$ for each $k \in K$, $\sum_{k \in K} w_k = 1$, and $\zeta \in [0,1]$.
\end{definition}

We use expressions $\esp$ to define distances between two \tracciones.
\emph{Atomic expressions} $\sx{\rho}$ and $\dx{\rho}$ are used to evaluate the distance between two distributions at a given time step with respect to the penalty function $\rho$.
Then we provide three \emph{temporal expression} operators, namely $\eventually{I}$, $\always{I}$ and $\until{I}$, allowing for the evaluation of minimal and maximal distances over time.
The $\meet,\join$ and \emph{convex combination} $\sum_{K} w_k$ operators allow us to evaluate the corresponding functions over expressions.
The \emph{comparison operator} $\sigma(\esp,\bowtie \zeta)$ returns a value in $\{0,1\}$ used to establish whether the evaluation of $\esp$ is in relation $\bowtie$ with the threshold $\zeta$.
Summarising, by means of expressions we can measure the differences between \tracciones{} with respect to various tasks (penalty functions) and temporal constraints.

Formulae are evaluated in an \traccione{} and a time instant.
The semantics of Boolean operators and bounded until is defined as usual, while the semantics of atomic propositions follows from the evaluation of expressions.

\begin{definition}
[Semantics of formulae]
\label{def:phi_semantics}
Let $\ES$ be an \traccione, and $\tau$ a time step.
The satisfaction relation $\models$ is defined inductively on the structure of formulae as:
\begin{itemize}
\item $\ES,\tau \models \top$ for all $\ES,\tau$;
\item $\ES,\tau \models \Delta(\esp,\p) \bowtie \eta$ if{f} $\sat{\esp}{\ES,\ES_{\mid_{\psem{\p},\tau}}}{\tau} \bowtie \eta$;
\item $\ES,\tau \models \neg \varphi$ if{f} $\ES,\tau \not\models \varphi$;
\item $\ES,\tau \models \varphi_1 \wedge \varphi_2$ if{f} $\ES,\tau \models \varphi_1$ and $\ES,\tau \models \varphi_2$;
\item $\ES,\tau \models \funtil{\varphi_1}{I}{\varphi_2}$ if{f} there is a $\tau' \in I+\tau$ s.t.\ $\ES,\tau' \models \varphi_2$ and for all $\tau'' \in I+\tau, \tau'' < \tau'$ it holds that $\ES,\tau'' \models \varphi_1$;  
\end{itemize}
where, for $I = [a,b]$, we let $I + \tau = [\min\{a+\tau,\h\},\min\{b+\tau,\h\}]$.
\end{definition}

Let us focus on atomic propositions.
We have that the \traccione{} $\ES$ at time $\tau$ satisfies the formula $\Delta(\esp,\p) \bowtie \eta$ if and only if the distance defined by $\esp$ between $\ES$ and $\ES_{\mid_{\psem{\p},\tau}}$ is $\bowtie \eta$, where $\ES_{\mid_{\psem{\p},\tau}}$ is the \traccione{} obtained by applying the perturbation $\p$ to $\ES$ at time $\tau$.
Formally, $\ES_{\mid_{\psem{\p},\tau}}$ is defined as:
\[
(\ES_{\mid_{\psem{\p},\tau}})^t = 
\begin{cases}
\ES^t & \text{ if } t < \tau, \\
\ES^{\psem{\p},t-\tau}_{\ES^\tau} & \text{ if } t \ge \tau.
\end{cases}
\]

Hence, for the first $\tau-1$ steps $\ES_{\mid_{\psem{\p},\tau}}$ is identical to $\ES$.
At time $\tau$ the perturbation $\p$ is applied, and the distributions in $\ES_{\mid_{\psem{\p},\tau}}$ are thus given by the perturbation via $\psem{\p}$ of the \traccione{} having $\ES^{\tau}$ as initial distribution (Definition~\ref{def:traccione_perturbato}).
It is worth noticing that by combining atomic propositions with temporal operators we can apply (possibly) different perturbations at different time steps, thus allowing for the analysis of systems behaviour in complex scenarios.

As expected, other operators can be defined as macros in our logic:
\[
\begin{array}{ll}
\varphi_1 \vee \varphi_2 \equiv \neg (\neg\varphi_1 \wedge \neg \varphi_2)
\qquad & \qquad
\varphi_1 \Longrightarrow \varphi_2 \equiv \neg\varphi_1 \vee \varphi_2 \\[.1cm]
\fevent{I}{\varphi} \equiv \funtil{\top}{I}{\varphi}
\qquad & \qquad
\fglob{I}{\varphi} \equiv \neg \fevent{I}{\neg \varphi}.
\end{array} 
\]

Finally, we present the evaluation of expressions.
As they define distances over \tracciones, which are sequences over time of distributions, they are evaluated over two \tracciones{} and a time $\tau$, representing the time step at which (or starting from which) the differences in the two sequences become relevant.

\begin{definition}
[Semantics of expressions]
\label{def:esp_semantics}
Let $\ES_1,\ES_2$ be to \tracciones, and $\tau$ be a time step.
The evaluation of expressions in the triple $\ES_1,\ES_2,\tau$ is the function $\sat{\cdot}{\ES_1,\ES_2}{\tau} \colon \espClass \to [0,1]$ defined inductively on the structure of expressions as follows:
\begin{itemize}
\item $\sat{\sx{\rho}}{\ES_1,\ES_2}{\tau} = \Wasserstein(m_{\rho,\tau})(\ES_{1}^{\tau},\ES_{2}^{\tau})$;
\item $\sat{\dx{\rho}}{\ES_1,\ES_2}{\tau} = \Wasserstein(m_{\rho,\tau})(\ES_{2}^{\tau},\ES_{1}^{\tau})$;
\item $\sat{\eventually{I} \esp}{\ES_1,\ES_2}{\tau} = \min_{t \in I + \tau} \sat{\esp}{\ES_1,\ES_2}{t}$;
\item $\sat{\always{I} \esp}{\ES_1,\ES_2}{\tau} = \max_{t \in I + \tau} \sat{\esp}{\ES_1,\ES_2}{t}$;
\item $\sat{\esp_1 \until{I} \esp_2}{\ES_1,\ES_2}{\tau} = \min_{t \in I + \tau} \max\{\sat{\esp_2}{\ES_1,\ES_2}{t}, \max_{t' \in I+\tau, t'< t} \sat{\esp_1}{\ES_1,\ES_2}{t'}\}$;
\item $\sat{\meet(\esp_1,\esp_2)}{\ES_1,\ES_2}{\tau} = \min\{\sat{\esp_1}{\ES_1,\ES_2}{\tau},\sat{\esp_2}{\ES_1,\ES_2}{\tau}\}$;
\item $\sat{ \join(\esp_1,\esp_2)}{\ES_1,\ES_2}{\tau} = \max\{\sat{\esp_1}{\ES_1,\ES_2}{\tau},\sat{\esp_2}{\ES_1,\ES_2}{\tau}\}$;
\item $\sat{\sum_{k \in K} w_k \esp_k}{\ES_1,\ES_2}{\tau} = \sum_{k \in K} w_k \cdot \sat{\esp_k}{\ES_1,\ES_2}{\tau}$;
\item $\sat{\sigma(\esp,\bowtie \zeta)}{\ES_1,\ES_2}{\tau} = 
\begin{cases}
0 & \text{ if } \sat{\esp}{\ES_1,\ES_2}{\tau} \bowtie \zeta, \\
1 & \text{ otherwise.}
\end{cases}$
\end{itemize}
\end{definition}

The evaluation bases on two atomic expressions, $\sx{\rho}$ and $\dx{\rho}$, where $\rho$ is a penalty function.
Given the \tracciones{} $\ES_1$, $\ES_2$, and a time $\tau$, we use $\sx{\rho}$ to measure the distance between the distributions reached by the two \tracciones{} at time $\tau$, i.e., $\ES^{1}_{\tau}$ and $\ES^{2}_{\tau}$, with respect to the penalty function $\rho$.
Formally, according to Definition~\ref{def:metric_prob_DS}, we use $\sx{\rho}$ to measure the distance $\Wasserstein(m_{\rho,\tau})(\ES_{1}^{\tau},\ES_{2}^{\tau})$.
Conversely, $\dx{\rho}$ measures the distance $\Wasserstein(m_{\rho,\tau})(\ES_{2}^{\tau},\ES_{1}^{\tau})$.
Having penalty functions as parameters of atomic expressions will allow us to study the differences in the behaviour of systems with respect to different data and objectives in time.
We remark that, in our setting, the need for two atomic expressions is justified by the choice of using a hemimetric to evaluate the distance between (distributions on) \datastates.
The temporal expression operators $\eventually{I},\always{I}$, and $\until{I}$ can be thought of as the quantitative versions of classical (bounded) temporal operators, respectively, \emph{eventually}, \emph{always}, and \emph{until}.
Their semantics follows by associating existential quantifications with minima, and universal quantifications with maxima.
Intuitively, asking for the existence of a time step at which a given formula is satisfied, corresponds to looking for a best case scenario, which in our quantitative setting can be matched with looking for the minimum over distances.
Similarly, when universal quantification is considered, the result can only be as good as the worst scenario, thus corresponding to the evaluation of the maximum over distances.
Hence, the evaluation of $\eventually{I}\esp$ is obtained as the minimum value of the distance $\esp$ over the time interval $I$.
Dually, $\always{I}\esp$ gives us the maximum value of $\esp$ over $I$.
Then, the evaluation of $\esp_1 \until{I} \esp_2$ follows from the classical semantics of bounded until (see Definition~\ref{def:phi_semantics}), accordingly. 
The evaluation of $\meet,\join$, and $\sum_K r_k$ is as expected.
The expression $\sigma(\esp,\bowtie \zeta)$ evaluates to $0$, i.e., the minimum distance, if the evaluation of $\esp$ is $\bowtie \zeta$; otherwise it evaluates to $1$, i.e., the maximum distance.
Informally, the comparison operator $\sigma$ can be combined with temporal expression operators to check whether several constraints of the form $\bowtie \zeta_i$ are satisfied over a time interval under a single application of a perturbation function (see Example~\ref{ex:different_temporal} below).


\paragraph{Robustness properties}

As an example of a robustness property that can be expressed in \logicShort, we consider the property of \emph{adaptability} presented in \cite{CLT21}.
Informally, we say that a system is adaptable if whenever its initial behaviour is affected by the perturbations, it is able to react to them and regain its intended behaviour within a given amount of time.
Given the thresholds $\eta_1,\eta_2 \in [0,1)$, an observable time $\tilde{\tau}$, and a penalty function $\rho$, we say that a system is ($\eta_1$, $\eta_2$, $\tilde{\tau}$)-adaptable with respect to $\rho$ if whenever a perturbation $\p$ occurring at time $0$ modifies the initial behaviour of the system by at most $\eta_1$, then the pointwise distance between the \tracciones{} of the two systems (the nominal one and its perturbation via $\p$) after time $\tilde{\tau}$ is bounded by $\eta_2$.
Hence, adaptability can be expressed in \logicShort{} as follows, where we recall that $\h$ is the finite time horizon for the evaluation of the property of interest:
\[
\Delta( \sx{\rho}, \p ) \le \eta_1
\Longrightarrow
\Delta( \always{[\tilde{\tau},\h]} \sx{\rho}, \p ) \le \eta_2.
\]
That of adaptability is just one simple example of a robustness property that can be expressed in \logicShort.
We now provide various examples of formulae that can be used for the analysis of the robustness of the engine system from Example~\ref{ex:engine_1}, and that should help the reader to further grasp the role of the operators in \logicShort.

\begin{example}
\label{ex:different_temporal}
Consider the penalty functions given in Example~\ref{ex:penalty}, and the perturbations from Example~\ref{ex:perturbation}.
We can build a formula $\varphi_1$ expressing that the attack on the insecure channel $\mathit{ch\_temp}$ is successful.
We recall that the attacker provides a negative temperature offset $o \sim \mathcal{U}[l_o,r_o]$ for $n$ time instants starting from the current instant $\tau$.
This attack is successful if, whenever the difference observed along the attack window $[\tau,\tau+n-1]$ between the physical value of temperature and that read by the controller is in the interval $[\eta_1 , \eta_2]$, for suitable values $\eta_1$ and $\eta_2$, then the level of the alarms raised by the IDS remains below a given stealthiness threshold $\eta_3$, 
and the level of system stress overcomes a danger threshold $\eta_4$ within $k$ unit of time:
\begin{align*}
\varphi_1 ={} & 
\fevent{[0,\h]}{(\varphi_1' \Longrightarrow \varphi_1'')} \\
\varphi_1' ={} &
\Delta(\eventually{[\tau,\tau+n-1]}\sx{\rho^{\mathit{temp}}},\p_{\mathit{temp},o,\tau,n}) \ge \eta_1
\wedge
\Delta(\always{[\tau,\tau+n-1]}\sx{\rho^{\mathit{temp}}},\p_{\mathit{temp},o,\tau,n}) \le \eta_2
\\
\varphi_1'' ={} & 
\Delta(\always{[\tau,k]}\sx{\rho^{\mathit{wrn}}}
,\p_{\mathit{temp},o,\tau,n}) \le \eta_3
\wedge
\Delta(\always{[\tau,k]}\sx{\rho^{\mathit{stress}}},\p_{\mathit{temp},o,\tau,n}) \ge \eta_4.
\end{align*}

Consider now the attack on actuator $\mathit{cool}$ described in Example~\ref{ex:perturbation}.
We provide a formula $\varphi_2$ expressing that such an attack fails within $k$ units of time.
This happens when the level of the alarms raised by the IDS goes above a given threshold $\zeta_2$ at some instant $\tau' \in [0,k]$, thus implying that the attack is detected, cowhile the level of stress remains below an acceptable threshold $\zeta_1$:
\begin{equation*}
\varphi_2 = \Delta
\left(
\sigma(\sx{\rho^{\mathit{stress}}}, < \zeta_1) 
\until{[0,k]}
\sigma(\sx{\rho^{\mathit{wrn}}}, > \zeta_2),
\p_{\mathit{cool},t,n}
\right)
< 1.
\end{equation*}
Notice that in the formula $\varphi_2$ the perturbation $\p_{\mathit{cool},t,n}$ is applied only once, at time $0$, and by means of the comparison and until operators on expressions we can evaluate all the distances along the considered interval between the original \traccione{} and its perturbation via $\p_{\mathit{cool},t,n}$.
Conversely, in the formula $\varphi_3$ below, the time step at which a perturbation is applied is determined by the bounded until operator:
\[
\varphi_3 = 
\funtil{\varphi_2}
{[\tau_1,\tau_2]}
{\Delta(\sx{\rho^{\mathit{fn}}}, \p_{\mathit{cool},t,n}) \le \eta_3}.
\]
The formula $\varphi_3$ is satisfied if there is a $\tilde{\tau} \in [\tau_1,\tau_2]$ such that:
\begin{inparaenum}
\item the attack on actuator $\mathit{cool}$ is detected  regardless of the time step in $[\tau_1,\tilde{\tau})$ at which $\p_{\mathit{cool},t,n}$ is applied, and 
\item the IDS is effective, up to tolerance $\eta_3$, against an application of $\p_{\mathit{cool},t,n}$ at time $\tilde{\tau}$.
\end{inparaenum}
We recall that the effectiveness of the IDS is measured in terms of the penalty function $\rho^{\mathit{fn}}$ on false negatives presented in Example~\ref{ex:penalty}.
\end{example}


\section{Statistical Evaluation of \logicName{} formulae}
\label{sec:statisticalmc}

In this section we outline the procedure, based on \emph{statistical techniques} and \emph{simulation}, that allows us to verify \logicShort{} specifications.

The proposed procedure  consists of the following basic steps:
\begin{enumerate}[(i)]
\item 
A randomised procedure that, based on simulation, permits the estimation of the \traccione{} of system $\system$, assuming an initial \datastate{} $\ds_\system$. Starting from $\ds_\system$ we have to sample $N$ sequences of \datastates{} $\ds_0^{j},\dots,\ds_k^{j}$, for $j=1,\dots,N$. All the \datastates{} collected at time $i$ are used to estimate the distribution $\ES_{\dirac_{\ds_\system}}^{i}$.
\item 
A procedure that given a \traccione{} permits to sample the effects of a perturbation. The same approach used to obtain an estimation of the \traccione{} associated with a given initial \datastate{} $\ds_\system$ can be used to obtain its perturbation.
The only difference is that while for \tracciones{} the \datastate{} $\ds_{i+1}$ at step $i+1$ only depends on the \datastate{} $\ds_i$ at step $i$, here the effects of a perturbation $\p$ are also applied. 
To guarantee statistical relevance of the collected data, for each sampled \datastate{} in the original \traccione{} we use an additional number $\ell$ of samplings to estimate the effects of $\p$ over it.
\item
A procedure that permits to estimate a distance expression between a \traccione{} and its perturbed variant. 
A distance expression $\exp$ is estimated following a syntax driven procedure.
To deal with the base case of atomic expressions $\sx{\rho}$ and $\dx{\rho}$, we rely on a mechanism to estimate the Wasserstein distance between two probability distributions on $(\D,\borel_\D)$.  
Following an approach similar to the one presented in~\cite{TK09}, to estimate the Wasserstein distance $\Wasserstein(m_{\rho,i})$ between (the unknown) $\mu$ and $\nu$ distributions on $(\D,\borel_{\D})$,
we can use $N$ independent samples $\{\ds^1_1,\ldots,\ds^N_1\}$ taken from $\mu$ and $\ell N$ independent samples $\{\ds^1_2,\ldots,\ds_2^{\ell N}\}$ taken from $\nu$. 
\item 
A procedure that checks if a given \traccione{} satisfies a formula $\varphi$. 
\end{enumerate}
These four steps are presented with more details in sections~\ref{subsec:evolseq}--\ref{subsec:checking}.

Since the procedures outlined above are based on statistical inference, we need to take into account the statistical error when checking the satisfaction of formulae.
Hence, in Section~\ref{subsec:bootstrap} we discuss a classical algorithm for the evaluation of confidence intervals in the evaluation of distances.
Then, in Section~\ref{subsec:3v}, we propose a three-valued semantics for \logicShort{} specifications, in which the truth value \emph{unknown} is added to true and false.
The three-valued semantics is generated by atomic propositions $\Delta(\esp,\p)\bowtie\eta$: if $\eta$ belongs to the confidence interval of the evaluation of $\esp$, then $\Delta(\esp,\p)\bowtie\eta$ evaluates to \emph{unknown}, since the validity of the relation $\bowtie \eta$ may depend on the particular samples obtained in the simulation.


\begin{figure}
\begin{subfigure}[t]{0.49\textwidth}
\small
\begin{algorithmic}[1]
\Function{Simulate}{$\ds_\system, N, k$} 
\State $i\gets 0$
\State $E_0 \gets  (\underbrace{\ds_\system,\ldots \ds_\system}_{N})$
\While{$i< k$}
\State $E_{i+1} \gets \emptyset$
\For{$\ds \in E_i$}
\State $E_{i+1} \gets \Call{SimStep}{\ds}, E_{i+1}$
\EndFor{}
\State $i \gets i+1$
\EndWhile
\State \Return $E_0,\ldots, E_k$
\EndFunction
\end{algorithmic}
\caption{Simulation of a \traccione{}.}
\label{fig:algo_simulate}
\end{subfigure}
\begin{subfigure}[t]{0.49\textwidth}
\small
\begin{algorithmic}[1]
\Function{SimPer}{$E_0,\ldots, E_k, \p, \tau, \ell$} 
\State $\forall i< \tau. E_i' \gets E_i$
\State $E_\tau' \gets \ell\cdot E_\tau$ 
\State $i\gets \tau$
\While{$i< k$}
\State $f \gets \mathsf{effect}(\p)$
\State $p \gets \mathsf{next}(\p)$
\For{$\ds \in E_{i}'$}
\State $\ds' \gets \Call{Sample}{f(\ds)}$
\State $E_{i+1}' \gets E_{i+1}', \Call{SimStep}{\ds'}$
\State $i \gets i+1$
\EndFor{}
\EndWhile
\State \Return $E_0',\ldots, E_k'$
\EndFunction
\end{algorithmic}
\caption{Computation of the effect of a perturbation.}
\label{fig:algo_perturbation}
\end{subfigure}
\caption{Algorithms for the simulation of \tracciones: original and perturbed.}
\end{figure}


\subsection{Statistical estimation of \tracciones{}}
\label{subsec:evolseq}

Given an initial \datastate{} $\ds_\system$ and two integers $N$ and $k$, we use a function $\Call{Simulate}{}$ (defined in Figure~\ref{fig:algo_simulate}) to obtain an \emph{empirical \traccione{}} of size $N$ and length $k$ starting from $\ds_\system$. 
Each $E_i$ in Figure~\ref{fig:algo_simulate} is a tuple $\ds_i^{1},\ldots,\ds_i^{N}$ of \datastates{} that are used to estimate the probability distribution of the \traccione{} from $\ds_\system$ at time $i$. 

The first tuple $E_0$ consists of $N$ copies of $\ds_\system$, while $E_{i+1}$ is obtained from $E_i$ by simulating one computational step from each element in $E_i$.
This step is simulated via a function $\Call{SimStep}{}$ that mimics the behaviour of function $\cstep(\ds)(\mesD)$ discussed in Section~\ref{sec:tracciones}. 
We assume that for any $\ds$ and for any measurable set $\mesD \in \borel_{\D}$ it holds that 
$
Pr\{ \Call{SimStep}{\ds} \in \mesD \}=\cstep(\ds)(\mesD)
$.
For any $i$, with $0\leq i\leq k$, we let $\hat{\ES}_{\dirac_{\ds_\system}}^{i,N}$ be the distribution such that for any measurable set $\mesD \in \borel_{\D}$ we have
$
\displaystyle{\hat{\ES}_{\dirac_{\ds_\system}}^{i,N}(\mesD)=\frac{|E_i \cap \mesD|}{N}}
$.
We can observe that, by applying the weak law of large numbers to the i.i.d.\ samples, we get that $\hat{\ES}_{\dirac_{\ds_\system}}^{i,N}$ converges weakly to $\ES^i_{\dirac_{\ds_\system}}$, namely
$
\displaystyle \lim_{N\rightarrow \infty}\hat{\ES}_{\dirac_{\ds_\system}}^{i,N} = \ES^i_{\dirac_{\ds_\system}}
$.

\begin{notation}
\label{not:E_is_traccione}
Since the estimated \traccione{} $\hat{\ES}_{\dirac_{\ds_\system}}^N$ is univocally determined by the tuple of sampled \datastates{} $E_0,\dots,E_k$ (each of size $N$) computed via the function $\Call{Simulate}{}$, we shall henceforth identify them, and refer to  $E_0,\dots,E_k$ as the estimated \traccione{} of size $N$.
\end{notation}


\subsection{Applying perturbations to \tracciones}

To compute the effect of a perturbation specified by $\p$ on an estimated \traccione{} $E_0,\dots,E_k$ of size $N$ (computed in procedure $\Call{Simulate}{}$), we use the function $\Call{SimPer}{}$ defined in Figure~\ref{fig:algo_perturbation}.

This function takes as parameters the time $\tau$ at which $\p$ is applied, and an integer $\ell$, giving the number of new samplings generated to amplify the effect of $\p$.
Given a sample set $E$, we let $\ell\cdot E$ denote the sample set obtained from $E$ by replicating each of its elements $\ell$ times.
The structure of this function is similar to that of function $\Call{Simulate}{}$.
However, while in the latter $E_{i+1}$ is obtained by applying the simulation step $\Call{SimStep}{}$, in $\Call{SimPer}{}$ the effect of a perturbation is first sampled. 
This is done by function $\Call{Sample}{f(\ds)}$, that can be defined in a standard way and that it is not reported here. 
According to semantics of perturbations given in Section~\ref{sec:perturbations}, the function $f$ used in $\Call{Sample}{f(\ds)}$ is $\mathsf{effect}(\p)$, and the perturbation used at next time step is $\mathsf{next}(\p)$.

\begin{figure}[t]
\small
\begin{algorithmic}[1]
\Function{ComputeWass}{$E_1,E_2,op,\rho$}
\State $(\ds^1_1,\ldots,\ds^N_1)\gets E_1$
\State $(\ds^1_2,\ldots,\ds_2^{\ell N})\gets E_2$
\State $\forall j: (1\leq j\leq N): \omega_j\gets\rho(\ds^j_1)$
\State $\forall h: (1\leq h\leq \ell N): \nu_h\gets\rho(\ds^h_2)$
\State re index $\{\omega_j\}$ s.t.\ $\omega_j\leq \omega_{j+1}$	
\State re index $\{\nu_h\}$ s.t.\ $\nu_h\leq \nu_{h+1}$
\If{$op=\sx{}$}
\State \Return $\frac{1}{\ell N}\displaystyle{\sum_{h=1}^{\ell N}\max\{ \nu_h - \omega_{\lceil \frac{h}{\ell}\rceil}, 0\}}$
\Else{}
\State \Return $\frac{1}{\ell N}\displaystyle{\sum_{h=1}^{\ell N}\max\{ \omega_{\lceil \frac{h}{\ell}\rceil}-\nu_h, 0\}}$
\EndIf{}
\EndFunction
\end{algorithmic}
\caption{Evaluation of the Wasserstein distance.}
\label{fig:algo_wasse}
\end{figure}


\subsection{Evaluation of distance expressions}


\paragraph{Statistical estimation of the Wasserstein metric}

Following an approach similar to the one presented in~\cite{TK09}, to estimate the Wasserstein distance $\Wasserstein(m_{\rho,i})$ between the (unknown) distributions $\mu$ and $\nu$, we can use $N$ independent samples $\{\ds^1_1,\ldots,\ds^N_1\}$ taken from $\mu$ and $\ell N$ independent samples $\{\ds^1_2,\ldots,\ds_2^{\ell N}\}$ taken from $\nu$. 
We then exploit the $i$-penalty function $\rho_i$ to map each sampled \datastate{} onto $\real$, so that, 
to evaluate the distance, it is enough to consider the reordered sequences of values 
$\{ \omega_j=\rho_i(\ds_1^j) \mid \omega_j\leq \omega_{j+1} \}$ 
and 
$\{ \nu_h=\rho_i(\ds_2^h) \mid \nu_{h}\leq \nu_{h+1} \}$.
The value $\Wasserstein(m_{\rho,i})(\nu,\mu)$ can be approximated as 
$\frac{1}{\ell N}\sum_{h=1}^{\ell N}\max\{\nu_{h} - \omega_{\lceil \frac{h}{M}\rceil},0\}$ \cite{CLT21}.

We let $\Call{ComputeWass}{}$, in Figure~\ref{fig:algo_wasse}, be the function that implements the procedure outlined above in order to estimate the distance expressions $\sx{\rho}$ or $\dx{\rho}$ between unknown distributions $\mu$ and $\nu$.
Note that, the third parameter of this function is the operator $op$ that can be either $\sx{}$ or $\dx{}$.
In the first case, the approximation of $\Wasserstein(m_{\rho,i})(\mu,\nu)$ is returned, while in the latter the one of $\Wasserstein(m_{\rho,i})(\nu, \mu)$.
We remark that the penalty function allows us to reduce the evaluation of the Wasserstein distance in $\real^n$ to its evaluation on $\real$.
Hence, due to the sorting of $\{\nu_h \mid h \in [1,\dots,\ell N]\}$ the complexity of outlined procedure is $O(\ell N \log(\ell N))$ (cf.\ \cite{TK09}).
We refer the interested reader to \cite[Corollary 3.5, Equation (3.10)]{SFGSL12} for an estimation of the approximation error given by the evaluation of the Wasserstein distance over $N,\ell N$ samples.


\paragraph{Evaluation of other expressions}

Function $\Call{EvalExpr}{}$, reported in Figure~\ref{alg:function_eval_esp}, can be used to evaluate a distance expression $\esp$ on the estimated \tracciones{} $E_0$, \ldots, $E_k$ and $E_0'$,\ldots, $E_k'$, computed by functions $\Call{Simulate}{}$ and $\Call{SimPer}{}$, respectively, at a given time $\tau$.
Function $\Call{EvalExpr}{}$ is defined recursively on the syntax of $\esp$ and follows the same scheme of Definition~\ref{def:esp_semantics}.

\begin{figure}
\small
\begin{algorithmic}[1]
\Function{EvalExpr}{$\{E_0,\ldots,E_{k}\},\{E_0',\ldots,E_{k}'\},\tau,\esp$} 
\Switch{$\exp$}
\Case{$\sx{\rho}$} 
\State \Return $\Call{ComputeWass}{E_\tau,E'_\tau,\sx{},\rho}$
\EndCase
\Case{$\dx{\rho}$} 
\State \Return $\Call{ComputeWass}{E_\tau,E'_\tau,\dx{},\rho}$
\EndCase
\Case{$\eventually{I}\esp$}
\State \Return $\min_{i\in \tau+I} \{ \Call{EvalExpr}{\{E_0,\ldots,E_{k}\},\{E_0',\ldots,E_{k}'\},i,\esp} \}$
\EndCase
\Case{$\always{I}\esp$}
\State \Return $\max_{i\in \tau+I} \{ \Call{EvalExpr}{\{E_0,\ldots,E_{k}\},\{E_0',\ldots,E_{k}'\},i,\esp} \}$
\EndCase
\Case{$\esp_1 \until{[\tau_1, \tau_2]} \esp_2$}
\State $\forall i\in [\tau+\tau_1, \tau+\tau_2]~~d^2_{i} \gets \Call{EvalExpr}{\{E_0,\ldots,E_{k}\},\{E_0',\ldots,E_{k}'\},i,\esp_2}$
\State $\forall j\in [\tau+\tau_1, \tau+\tau_2]~~d^1_{j} \gets \Call{EvalExpr}{\{E_0,\ldots,E_{k}\},\{E_0',\ldots,E_{k}'\},j,\esp_1}$
\State \Return $\min_{~\tau+\tau_1\leq i\leq \tau+\tau_2} \{ 
\max \{ d^2_{i},
\max_{~0\leq j< i} \{ d^1_{j}\}\}$
\EndCase
\Case{$\meet (\esp_1, \esp_2)$} 
\State $v_1\gets \Call{EvalExpr}{\{E_0,\ldots,E_{k}\},\{E_0',\ldots,E_{k}'\},\tau,\esp_1}$
\State $v_2\gets \Call{EvalExpr}{\{E_0,\ldots,E_{k}\},\{E_0',\ldots,E_{k}'\},\tau,\esp_2}$
\State \Return $\min\{ v_1, v_2\}$
\EndCase
\Case{$\join (\esp_1, \esp_2)$} 
\State $v_1\gets \Call{EvalExpr}{\{E_0,\ldots,E_{k}\},\{E_0',\ldots,E_{k}'\},\tau,\esp_1}$
\State $v_2\gets \Call{EvalExpr}{\{E_0,\ldots,E_{k}\},\{E_0',\ldots,E_{k}'\},\tau,\esp_2}$
\State \Return $\join\{ v_1, v_2\}$
\EndCase
\Case{$\sum_{i \in K} w_i \cdot \esp_i$}
\State $v_i\gets \Call{EvalExpr}{\{E_0,\ldots,E_{k}\},\{E_0',\ldots,E_{k}'\},\tau,\esp_i}$
\State \Return $\sum_{i\in K} w_i\cdot v_i$
\EndCase
\Case{$\sigma(\esp,\bowtie \zeta)$}
\State $v\gets \Call{EvalExpr}{\{E_0,\ldots,E_{k}\},\{E_0',\ldots,E_{k}'\},\tau,\esp}$
\If{$v\bowtie \zeta$}
\State \Return $0$
\Else{}
\State \Return $1$
\EndIf{}
\EndCase
\EndSwitch
\EndFunction
\end{algorithmic}	
\caption{Evaluation of distance expressions.}
\label{alg:function_eval_esp}
\end{figure}


\subsection{Checking formulae satisfaction}
\label{subsec:checking}

To check if a given system $\system$ satisfies or not a given formula $\varphi$, starting from the \datastate{} $\ds_\system$, function $\Call{Sat}{}$, defined in Figure~\ref{alg:computesat}, can be used. 
Together with the \datastate{} $\ds_\system$ and the formula $\varphi$, function $\Call{Sat}{}$ takes as parameters the two integers $\ell$ and $N$ identifying the number of samplings that will be used to estimate the Wasserstein metric. 
This function consists of three steps. 
First the \emph{time horizon} $k$ of the formula $\varphi$ is computed (by induction on the structure of $\varphi$) to identify the number of steps needed to evaluate the satisfaction of the formula. 
In the second step, function $\Call{Simulate}{}$ is used to simulate the \traccione{} of $\system$ from $\ds_\system$ by collecting the sets of samplings $\overline{E} = E_0,\ldots,E_k$ needed to check the satisfaction of $\varphi$, in the third step, by calling function $\Call{Eval}{}$ defined in Figure~\ref{alg:function_eval}.

The structure of $\Call{Eval}{}$ is similar to the monitoring function for STL defined in~\cite{MN04}.
Given $N$ sampled values at time $0,\ldots,k$, a formula $\varphi$ and integers $\ell$, $\tau$, function $\Call{Eval}{}$ yields a boolean value indicating if $\varphi$ is satisfied or not at time step $\tau$. 

\begin{figure}
\begin{subfigure}[t]{0.45\textwidth}
\small
\begin{algorithmic}[1]
\Function{Sat}{$\ds_\system,\varphi,\ell,N$} 
\State $k\gets \Call{Horizon}{\varphi}$
\State $\overline{E} \gets  \Call{Simulate}{\ds_\system,k,N}$
\State \Return \Call{Eval}{$\overline{E},\ell,i,\varphi$}
\EndFunction
\end{algorithmic}
\caption{Checking the satisfaction of a formula.}
\label{alg:computesat}
\end{subfigure}
\begin{subfigure}[t]{0.54\textwidth}
\small
\begin{algorithmic}[1]
\Function{Eval}{$\overline{E},\ell, \tau,\varphi$} 
\Switch{$\varphi$}
\Case{$\varphi=\top$}
\State \Return $true$
\EndCase
\Case{$\varphi=\Delta(\esp,\p) \bowtie \eta$}
\State $\overline{E'} \gets \Call{SimPer}{\overline{E},\p, \tau,\ell}$
\State $v \gets \Call{EvalExpr}{\overline{E}, \overline{E'}, \tau, \esp}$
\State \Return $v \bowtie \eta$
\EndCase
\Case{$\varphi=\neg\varphi_1$}
\State \Return $\neg \Call{Eval}{\overline{E},\ell,\tau,\varphi_1}$ 
\EndCase
\Case{$\varphi_1 \wedge \varphi_2$}
\State \Return $\Call{Eval}{\overline{E},\ell, \tau,\varphi_1}\wedge \Call{Eval}{\overline{E},\ell, \tau,\varphi_2}$
\EndCase
\Case{$\funtil{\varphi_1}{[\tau_1, \tau_2]}{\varphi_2}$}
\State $res \gets true$
\State $i\gets \tau+\tau_2$
\While{$i>0$}
\State $res \gets res\wedge \Call{Eval}{\overline{E},\ell, i,\varphi_1}$
\If{$i\in [\tau+\tau_1, \tau+\tau_2]$}
\State $res \gets res\vee \Call{Eval}{\overline{E},\ell, i,\varphi_2}$
\EndIf{}
\EndWhile{}
\State \Return $res$ 
\EndCase
\EndSwitch{}
\EndFunction
\end{algorithmic}	
\caption{Evaluation of \logicShort{} formulae.}
\label{alg:function_eval}
\end{subfigure}
\caption{Model checking.}
\end{figure}

\begin{figure*}[t]
\begin{subfigure}{.48\textwidth}
\includegraphics[scale=0.4]{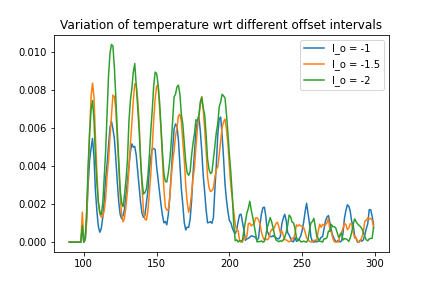}
\caption{Difference with respect to $\mathit{temp}$.}
\label{fig:distance_comparison_temp}
\end{subfigure}\hfill
\begin{subfigure}{.48\textwidth}
\includegraphics[scale=0.4]{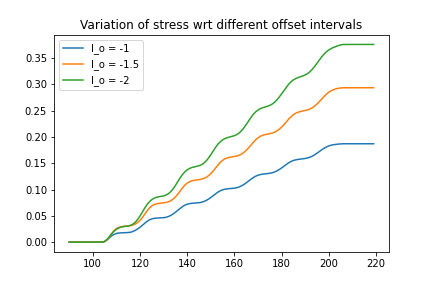}
\caption{Difference with respect to $\mathit{stress}$.}
\label{fig:distance_comparison_stress}
\end{subfigure}
\caption{Variation of the differences with respect to the values of $\mathit{temp}$ and $\mathit{stress}$, via $\p_{\mathit{temp},o,100,100}$ under different offset intervals ($l_o \in\{-2,-1.5,-1\}$).}
\label{fig:distance_comparison}
\end{figure*}

\begin{figure*}[t]
\begin{subfigure}{.48\textwidth}
\includegraphics[scale=0.4]{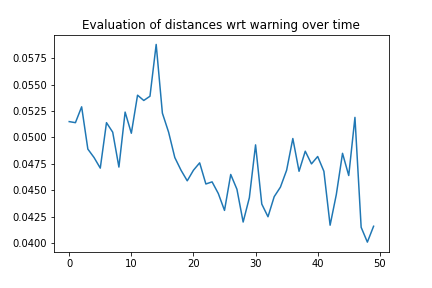}
\caption{$\esp_1$.}
\label{fig:eventually_wrn}
\end{subfigure}\hfill
\begin{subfigure}{.48\textwidth}
\includegraphics[scale=0.4]{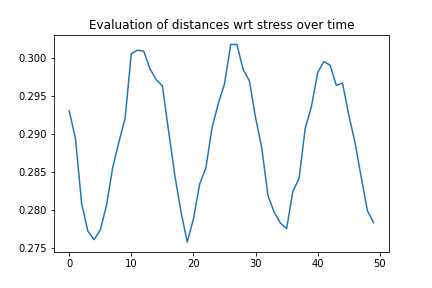}
\caption{$\esp_2$.}
\label{fig:eventually_stress}
\end{subfigure}
\caption{Variation of the evaluation of $\esp_1$ and $\esp_2$ over the time interval $[0,50]$.}
\label{fig:eventually}
\end{figure*}

\begin{example}
\label{ex:evaluations}
We give some examples of an application of our algorithms to the analysis of distances over the engine system.
We focus on the attack on the sensor $\mathit{temp}$, which is modelled by means of the perturbation $\p_{\mathit{temp},o,\tau,n}$ from Example~\ref{ex:perturbation}.
Let $\tau,n =100$, meaning that the atomic perturbation $\f_{\mathit{temp},o}$ is applied for the first time after $100$ steps from the current instant, and it is iterated for $100$ times.
For simplicity, let $0$ be the current instant.
To see the effects of $\p=\p_{\mathit{temp},o,100,100}$ over the \traccione{} $\ES$ of the engine, in Figure~\ref{fig:distance_comparison_temp} we report the pointwise evaluation of the expression $\sx{\rho}$, where $\rho(\ds) = \ds(\mathit{temp})/150$ for all $\ds \in \D$, over the time window $[90,300]$, giving thus the variation of the difference between the temperature in the perturbed system and that in the original one.
Clearly, this difference is greater in $[100,200]$, i.e., while $\f_{\mathit{temp},o}$ is active, and the smaller differences detected after $200$ steps are due to the delays induced by the perturbation in the regular behaviour of the engine.
To give a better idea of the impact of perturbations, in Figure~\ref{fig:distance_comparison_temp} we actually depicted the evaluation of these distances with respect to three variations of $\p$: we changed the left bound $l_o$ of the interval $[l_o,0]$ over which the offset $o$ is selected via a uniform distribution, $l_o \in \{-2,-1.5,-1\}$.
As one can expect, the larger the offset interval, the greater the difference.
This is even more evident in Figure~\ref{fig:distance_comparison_stress}, where we report the pointwise evaluations of the distances $\sx{\rho^{\mathit{stress}}}$ between $\ES$ and the the three perturbed \tracciones{}, for $\rho^{\mathit{stress}}$ defined in Example~\ref{ex:penalty} (results obtained with $\mathsf{stressincr} = 0.02$).

Let us now fix $l_o = -1.5$.
Consider the expressions
\[
\esp_1 = 
\always{[100,210]}\sx{\rho^{\mathit{wrn}}}
\qquad
\esp_2 =
\always{[100,210]}\sx{\rho^{\mathit{stress}}}
\]
which are instances of the expressions $\always{[\tau,k]}\sx{\rho^{\mathit{wrn}}}$ and $\always{[\tau,k]}\sx{\rho^{\mathit{stress}}}$ appearing in the formula $\varphi_1''$ in Example~\ref{ex:different_temporal}.
In Figure~\ref{fig:eventually} we report the variation of the evaluation of the two expressions over $\ES$ and its 51 perturbations via $\p$, each obtained by applying $\p$ at ad different instant $\tau' \in [0,50]$. 
For each step $\tau' \in [0,50]$, the interval over which the maxima of atomic expression $\sx{\rho^{\mathit{wrn}}}$ and $\sx{\rho^{\mathit{stress}}}$ are evaluated is $[100,210] + \tau'$.
Specifically, in the plots, we associate the coordinate $x=\tau'$ with the value $\sat{\esp_1}{\ES,\ES_{\mid_{\psem{\p},\tau'}}}{\tau'}$ in Figure~\ref{fig:eventually_wrn}, and with the value $\sat{\esp_2}{\ES,\ES_{\mid_{\psem{\p},\tau'}}}{\tau'}$ in Figure~\ref{fig:eventually_stress}.
The two plots show that by applying the perturbation at different time steps, we can get different effects on systems behaviour, with variations of the order of $10^{-3}$ in the case of warnings, and of the order of $10^{-2}$ in the case of stress.

We run several experiments in order to understand for which stealthiness threshold $\eta_3$ and danger threshold $\eta_4$ the formula $\varphi_1$ in Example~\ref{ex:different_temporal} expressing that the attack on the insecure channel $\mathit{ch\_temp}$ is successful is satisfied.
We concluded that for $\eta_3$ at least $0.06$ and $\eta_4$ at most $0.4$ the attack is successful (see Example~\ref{ex:three_valued} below for a further discussion on the tuning of $\eta_3$ and $\eta_4$).
\end{example}


\subsection{Statistical error}
\label{subsec:bootstrap}

In this section we discuss the evaluation of the statistical error arising
in the estimation of distances between a real \traccione{} and its perturbation, via the application function $\Call{EvalExpr}{}$.
More precisely, given a distance expression $\exp$, a real \traccione{} $\ES$, a perturbation $\p$, and two time instants $\tau$ and $\tau'$, our aim is to provide a procedure for the evaluation of a confidence interval $\cinterval$ such that the probability that the real value $\sat{\esp}{\ES,\ES_{\mid_{\psem{\p},\tau}}}{\tau'}$ of the distance is in $\cinterval$ is at least $\alpha$, for a desired coverage probability $\alpha$, i.e.,
$
Pr \{\sat{\esp}{\ES,\ES_{\mid_{\psem{\p},\tau}}}{\tau'} \in \cinterval\} \ge \alpha
$.

Notice that the statistical errors arising from the estimation of the real distributions in the \tracciones{} $\ES$, and $\ES_{\mid_{\psem{\p},\tau}}$, through their simulations via functions $\Call{Simulate}{}$, and $\Call{SimPer}{}$, are subsumed in the approximation errors on the evaluation of the distances.

We start from the evaluation of the confidence intervals on Wasserstein distances, i.e., a confidence interval $\cinterval$ such that
$
Pr \{\Wasserstein(m_{\rho})(\ES_{1}^{\tau},\ES_{2}^{\tau}) \in \cinterval\} \ge \alpha
$,
where $\ES_{1}^{\tau},\ES_{2}^{\tau}$ are the real distributions reached at time $\tau$ in the \tracciones, and $\alpha$ is the desired coverage probability.

As $\ES_1^{\tau}$ and $\ES_2^{\tau}$ are unknown, and, thus, so is the real value of the Wasserstein distance, to obtain an estimation of $\cinterval$ we apply the \emph{normal-theory intervals} (or \emph{empirical bootstrap}) method \cite{Ef79,Ef81}:
\begin{enumerate}
\item Generate $m$ bootstrap samples for each distribution: these are obtained by drawing with replacement a sample of size $N$ from the elements of the original sampling of $\mu$, and one of size $\ell N$ from those for $\nu$.
Let $\mu_1,\dots,\mu_m$ and $\nu_1,\dots,\nu_m$ the obtained bootstrap samples.
\item Apply the procedure $\Call{ComputeWass}{}$ $m$-times to evaluate the Wasserstein distances between the bootstrap distributions.
Let $W_1,\dots,W_m$ be the resulting bootstrap distances.
\item Evaluate the mean of the bootstrap distance
$ \displaystyle
\overline{W} = \frac{\sum_{i=1}^m W_i}{m}.
$
\item Evaluate the bootstrap estimated standard error
$ \displaystyle
SE_W = \sqrt{\frac{\sum_{i=1}^m (W_i - \overline{W})^2}{m-1}}.
$
\item Let
$ \displaystyle
\cinterval = \overline{W} \pm z_{1-\alpha/2} SE_W
$,
where $z_{1-\alpha/2}$ is the $1-\alpha/2$ quantile of the standard normal distribution.
\end{enumerate}

\begin{remark}
In \cite{TK08} a similar procedure is proposed, but there the bootstrap percentile interval method is used.
We chose to use the empirical bootstrap method to find a balance between accuracy and computational complexity.
In fact, it is known that empirical bootstraps are subject to bias in the samples, and some more accurate techniques, like the \emph{bias-corrected, accelerated percentile intervals} ($BC_a$), have been proposed \cite{Ciccio96}.
However, in order to reach the desired accuracy with the $BC_a$ method, it is necessary to use a number of bootstrap sampling $m \ge O(1000)$.
This means that in order to obtain an estimation of the confidence interval of a \emph{single evaluation} of a Wasserstein distance, we need to evaluate it at least $1000$ times \cite{FW18}.
Given the cost of a single evaluation, and considering that in our formulae this distance is evaluated thousands of times, this approach would be computationally unfeasible.
In our examples, a number of bootstrap samplings $m \le 100$ is sufficient to obtain reasonable confidence intervals (the width of our $95\%$ intervals is $O(10^{-3})$).
\end{remark}

The evaluation of the confidence interval for the computation of the Wasserstein distance is then extended to distance expressions: once we have determined the bounds of the confidence intervals of the sub-expressions occurring in $\esp$, the $\cinterval$ of $\esp$, denoted by $\cinterval_\esp$, is obtained by applying the function defining the evaluation of $\esp$ to them.
For instance, if $\esp = \join(\esp_1,\esp_2)$, $\cinterval_{\esp_1} = (l_1,r_1)$, and $\cinterval_{\esp_2} = (l_2,r_2)$, then $\cinterval_{\esp} = (\max\{l_1,l_2\}, \max\{r_1,r_2\})$.

\begin{example}
\label{ex:bootstrapping}
In Figure~\ref{fig:ex_ci} we report the $95\%$ confidence intervals for $\sat{\esp_1}{\ES,\ES_{\mid_{\psem{\p},0}}}{\tau'}$, where $\tau' \in [0,50]$, with $\esp_1$ and $\p$ as in Example~\ref{ex:evaluations}.
We remark that here the perturbation $\p$ is applied only once, at time $0$ (in fact we consider the \traccione{} $\ES_{\mid_{\psem{\p},0}}$).  
The intervals in Figure~\ref{fig:ex_ci_50} have been obtained by means of $m=50$ bootstrap samplings, whereas for those in Figure~\ref{fig:ex_ci_100} we used $m=100$.
In the former case, the maximal width of the interval is $9.57\cdot 10^{-3}$, with an average width of $8.03 \cdot 10^{-3}$; in the latter case, those number become, respectively, $9.48 \cdot 10^{-3}$ and $8.38 \cdot 10^{-3}$.
As the widths of the intervals are of the same order, we can limit ourselves to use $m=50$ in the experiments, thus lowering the computation time without loosing information.
\end{example}

\begin{figure}
\begin{subfigure}{.48\textwidth}
\includegraphics[scale=0.35]{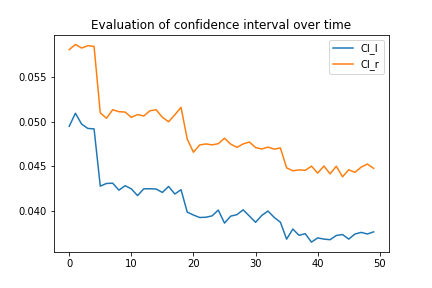}
\caption{$m=50$}
\label{fig:ex_ci_50}
\end{subfigure}
\begin{subfigure}{.48\textwidth}
\includegraphics[scale=0.35]{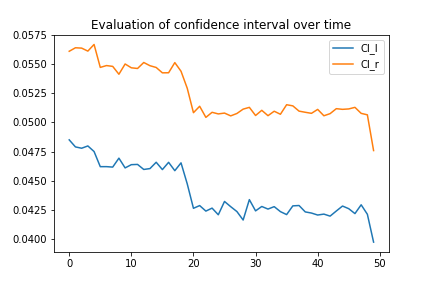}
\caption{$m=100$}
\label{fig:ex_ci_100}
\end{subfigure}
\caption{Confidence intervals of $\esp_1$, for $\alpha=0.05$, evaluated over the time interval $[0,50]$.}
\label{fig:ex_ci}
\end{figure}


\subsection{A three-valued semantics for \logicShort}
\label{subsec:3v}

The presence of errors in the evaluation of expressions due to statistical approximations has to be taken into account also when checking the satisfaction of \logicShort{} formulae.
Specifically, our model checking algorithm will assign a three-valued semantics to formulae by adding the truth value \emph{unknown} ($\Cup$) to the classic true ($\top$) and false ($\bot$).
Intuitively, unknown is generated by the comparison between the distance and the chosen threshold in atomic propositions: if the threshold $\eta$ does not lie in the confidence interval of the evaluation of the distance, then the formula will evaluate to $\top$ or $\bot$ according to whether the relation $\bowtie \eta$ holds or not.
Conversely, if $\eta$ belongs to the confidence interval, then the atomic proposition evaluates to $\Cup$, since the validity of the relation $\bowtie \eta$ may depend on the particular samples obtained in the simulation.

Starting from atomic propositions, the three-valued semantics is extended to the Boolean operators via truth tables in the standard way \cite{Kl52}.
Then, we assign a three-valued semantics to \logicShort{} formulae via the satisfaction function $\Omega_\ES \colon \logicSymbol \times [0,\h] \to \{\top,\Cup,\bot\}$, defined for all \tracciones{} $\ES$ as follows:
\begin{align*}
\Omega_\ES(\top,\tau) ={} &
\top
\\
\Omega_\ES(\Delta(\esp,\p) \bowtie \eta,\tau) ={} &
\begin{cases}
\Cup & \text{ if } \eta \in \cinterval_{\esp}
\\
\models(\ES,\tau,\Delta(\esp,\p) \bowtie \eta) & \text{ otherwise.}
\end{cases}
\\
\Omega_\ES(\neg\varphi,\tau) ={} & 
\neg \Omega_\ES(\varphi,\tau) 
\\
\Omega_\ES(\varphi_1 \wedge \varphi_2,\tau) ={} &
\Omega_\ES(\varphi_1,\tau) \wedge \Omega_\ES(\varphi_2,\tau)
\\
\Omega_\ES(\funtil{\varphi_1}{I}{\varphi_2},\tau) ={} &
\bigvee_{\tau' \in I} \left(\Omega_\ES(\varphi_2,\tau') 
\wedge
\bigwedge_{\tau'' \in I, \tau'' < \tau'} \Omega_{\ES}(\varphi_1,\tau'')\right).
\end{align*}

\begin{example}
\label{ex:three_valued}
Consider the formula $\varphi_{\eta_3}= \Delta(\esp_1,\p) \le \eta_3$ for $\esp_1$ and $\p$ as in Example~\ref{ex:evaluations}.
In Figure~\ref{fig:ex_three_valued} we report the variation of the evaluation of $\Omega_{\ES}(\varphi_{\eta_3},\tau')$ with respect to $\tau' \in [0,50]$ and $\eta_3 \in \{0.03,0.04,0.06\}$, where we let $\top \mapsto 1$, $\Cup \mapsto 0$, and $\bot \mapsto -1$.
The plot confirms the validity of the empirical tuning of parameter $\eta_3$ that we carried out in Example~\ref{ex:evaluations}.
A similar analysis (non reported here) has been conducted for $\eta_4$.
\end{example}

\begin{figure}
\includegraphics[scale=0.4]{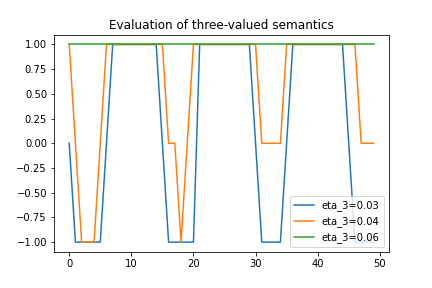}
\caption{Three-valued evaluation of the formula $\varphi_{\eta_3}$ over the time interval $[0,50]$ for $\eta_3 = 0.03, 0.04, 0.06$.}
\label{fig:ex_three_valued}
\end{figure}


\section{Concluding remarks}
\label{sec:conclusion}

We have introduced the \logicName{} (\logicShort), a temporal logic allowing for the specification and analysis of properties of distances between the behaviours of CPSs over a finite time horizon.
Specifically, we have argued that the unique features of \logicShort{} make it suitable for the verification of robustness properties of CPSs against perturbations.
Moreover, it also allows us to capture properties of the probabilistic transient behaviour of systems.

As briefly discussed in the introduction, the term robustness is used in several contexts, from control theory \cite{ZD98} to biology \cite{Ki07}, and not always with the same meaning.
Since our objective was not to introduce a notion of robustness, but to provide a formal tool for the verification of general robustness properties, we limit ourselves to recall that, in the context of CPSs, we can distinguish five categories of robustness \cite{FKP16}:
\begin{inparaenum}[(i)]
\item input/output robustness;
\item robustness with respect to system parameters;
\item robustness in real-time system implementation;
\item \label{unpred-env} robustness due to unpredictable environment;
\item robustness to faults.
\end{inparaenum}
Our framework is designed for properties falling into the (\ref{unpred-env}) category.
An interesting avenue for future research is to check whether robustness properties from the other categories can be specified using our framework.

Up to our knowledge, \logicShort{} is the only existing temporal logic expressing properties of distances between systems behaviours.
Usually, even in logics equipped with a real-valued semantics (which, in an unfortunate twist of fate, is also known as robustness semantics), the behaviour of a given system is compared to the desired property.
Moreover, specifications are tested only over a single trajectory of the system at a time.
Here we are interested in studying how the distance between two systems evolves in time, and to do that we always take into account the overall behaviour of the system, i.e., \emph{all possible trajectories}.
This feature also distinguishes our approach to robustness from classical ones, like those in~\cite{FP09,DM10}.
Our properties are based on the comparison of the \tracciones{} of two different systems, whereas \cite{FP09,DM10} compare a single behaviour of a single system with the set of the behaviours that satisfy a given property, which is specified by means of a formula expressed in a suitable temporal logic.

Recently, \cite{WRWVD19} proposed a statistical model checking algorithm based on stratified sampling for the verification of PCTL specification over Markov chains.
Informally, stratified sampling allows for the generation of negatively correlated samples, i.e., samples whose covariance is negative, thus considerably reducing the number of samples needed to obtain confident results from the algorithm.
However, the proposed algorithm works under a number of assumptions restricting the form of the PCTL formulae to be checked.
While direct comparison of the two algorithms would not be feasible, nor meaningful given the disparity in the classes of formulae, it would be worth studying the use of stratified sampling in our model checking algorithm.

As another direction for future work, we plan to develop a predictive model for the runtime monitoring of \logicShort{} specifications.
In particular, inspired by \cite{PPZGSS18,BCPSS19} where deep neural networks are used as reachability predictors for predictive monitoring, we intend to integrate our work with learning techniques, to favour the computation and evaluation of the predictions.

We also plan to apply our framework to the analysis of biological systems.
Some quantitative extensions of temporal logics have already been proposed in that setting (e.g. \cite{FR08,RBFS09,RBFS11}) to capture the notion of robustness from \cite{Ki07} or similar proposals \cite{NGM18}.
It would be interesting to see whether the use of \logicShort{} and \tracciones{} can lead to new results.

Finally, we will investigate the application of our framework to Medical CPSs.
In this context, statistical inference and learning methods have been combined in the synthesis of controllers, in order to deal with uncertainties (see, e.g., \cite{PLCSL20}).
The idea is then to use our tool to test the obtained controllers and verify their robustness against uncertainties.

\bibliographystyle{ACM-Reference-Format}
\bibliography{RobTL}

\end{document}